\documentclass[letterpaper]{article} 
\usepackage{aaai2026}  
\usepackage{times}  
\usepackage{helvet}  
\usepackage{courier}  
\usepackage[hyphens]{url}  
\usepackage{graphicx} 
\usepackage{natbib}  
\usepackage{caption} 
\usepackage{algorithm}
\usepackage{algorithmic}

\usepackage{newfloat}
\usepackage{listings}
\DeclareCaptionStyle{ruled}{labelfont=normalfont,labelsep=colon,strut=off} 
\floatstyle{ruled}
\newfloat{listing}{tb}{lst}{}
\floatname{listing}{Listing}
\nocopyright 

\newif\iflong
\newif\ifshort

\longtrue

\iflong
\else
\shorttrue
\fi

\usepackage[utf8]{inputenc} 
\usepackage[T1]{fontenc}    
\usepackage{url}            
\usepackage{booktabs}       
\usepackage{amsfonts}       
\usepackage{nicefrac}       
\usepackage{microtype}      
\usepackage{commath}
\usepackage{thm-restate}

\usepackage{graphicx}
\usepackage{amssymb}
\usepackage{mathtools}
\usepackage{enumitem}
\usepackage[amsmath, thmmarks]{ntheorem}
\usepackage{xspace}
\usepackage[capitalise]{cleveref}

\usepackage{tabularx}
\usepackage{environ}
\usepackage{xpatch}
\usepackage{mdframed}

\theoremstyle{plain}
\theorempostskip{-2pt}
\theoremheaderfont{\upshape\bfseries}
\theorembodyfont{\itshape}
\theoremseparator{.}

\newtheorem{lemma}{Lemma}
\newtheorem{corollary}{Corollary}

\newtheorem{observation}{Observation}

\newtheorem{claim*}{Claim}

\newtheorem{longtheorem}{Theorem}
\newtheorem{longlemma}[longtheorem]{Lemma}

 \theoremstyle{nonumberplain}
 \theoremheaderfont{\itshape}
 \theorembodyfont{\normalfont}
  \theoremseparator{.}
 \theoremsymbol{}
 \theoremsymbol{\ensuremath\square}
 \newtheorem{proof}{Proof}
 \newtheorem{proofsketch}{Proof Sketch}
 \newtheorem{claimproof}{Proof of the Claim}
 
\newcommand{\yesinstance}{\textup{YES}-instance\xspace}

\newcommand{\noinstance}{\textup{NO}-instance\xspace}

\newcommand{\bigO}[1]{{\ensuremath{\mathcal{O}\!\left(#1\right)}}\xspace}
\newcommand{\bigoh}{\mathcal{O}}

\newcommand{\N}{\mathbb{N}}

\newcommand{\calE}{{\cal E}}

\newcommand{\calV}{{\cal V}}

\newcommand{\calP}{{\cal P}}
\newcommand{\calQ}{{\cal Q}}

\newcommand{\calR}{{\cal R}}

\newcommand{\para}[1]{\smallskip \noindent \textbf{#1}\quad}

\newcommand{\fpt}{\textup{FPT}\xspace}
\newcommand{\XP}{\textup{XP}\xspace}

\newcommand{\wone}{\textup{W[1]}\xspace}

\newcommand{\NP}{\textup{NP}\xspace}
\newcommand{\np}{\NP\xspace}

\newcommand{\nphard}{\np-hard\xspace}

\newcommand{\nphardness}{\np-hardness\xspace}

\newcommand{\wonehard}{\wone-hard\xspace}

\newcommand{\labeltop}[1]{{\text{\makebox[0pt][c]{#1}}}}

\newcommand{\editcost}[2]{\norm{{#1} - {#2}}_{\neq 0}}

\newcommand{\descendants}{\widetilde{\chi}}

\newcommand{\egr}[1]{G^{#1}}

\makeatletter
\newcommand{\problemtitle}[1]{\gdef\@problemtitle{#1}}
\newcommand{\probleminput}[1]{\gdef\@probleminput{#1}}
\newcommand{\problemtask}[1]{\gdef\@problemtask{#1}}
\NewEnviron{problem}{
\problemtitle{}\probleminput{}\problemtask{}
\BODY
\noindent
\begin{tabularx}{\textwidth}{
	 						l X c}
	\multicolumn{2}{
						l}{\@problemtitle} \\
	\textbf{Input:} & \@probleminput \\
	\textbf{Task:} & \@problemtask
\end{tabularx}
}
\makeatother

\makeatletter
\xpatchcmd{\endmdframed}
  {\aftergroup\endmdf@trivlist\color@endgroup}
  {\endmdf@trivlist\color@endgroup\@doendpe}
  {}{}
\makeatother

\surroundwithmdframed[linecolor=black,linewidth=1pt,skipabove=1mm,skipbelow=1mm,leftmargin=0mm,rightmargin=0mm,innerleftmargin=1mm,innerrightmargin=1mm,innertopmargin=1mm,innerbottommargin=1mm]{problem}

\definecolor{cb_orange}{rgb}{0.859 0.427 0.0}
\definecolor{cb_blue}{rgb}{0.0 0.427 0.859}
\definecolor{cb_violet}{rgb}{0.714 0.427 1.0}
\definecolor{cb_dark_seal}{rgb}{0.0 0.286 0.286}
\definecolor{cb_seal}{rgb}{0.0 0.573 0.573}
\definecolor{cb_pink}{rgb}{1.0 0.427 0.714}
\definecolor{cb_rose}{rgb}{1.0 0.714 0.859}
\definecolor{cb_purple}{rgb}{0.286 0.0 0.573}
\definecolor{cb_light_blue}{rgb}{0.427 0.714 1.0}
\definecolor{cb_vlight_blue}{rgb}{0.714 0.859 1.0}
\definecolor{cb_red}{rgb}{0.573 0.0 0.0}
\definecolor{cb_brown}{rgb}{0.573 0.286 0.0}
\definecolor{cb_green}{rgb}{0.141 1.0 0.141}
\definecolor{cb_yellow}{rgb}{1.0 1.0 0.427}

\definecolor{cb2_orange}{rgb}{0.961, 0.475, 0.227}
\definecolor{cb2_purple}{rgb}{0.663, 0.353, 0.631}
\definecolor{cb2_cyan}{rgb}{0.522, 0.753, 0.976}
\definecolor{cb2_blue}{rgb}{0.059, 0.125, 0.502}

\DeclareMathOperator{\dr}{dr}
\DeclareMathOperator{\tw}{tw}
\DeclareMathOperator{\type}{\mathcal{T}}
\DeclareMathOperator{\col}{col}

\DeclareMathOperator{\cost}{cost}

\DeclareMathOperator{\MM}{\ensuremath{\mathbf{M}}}
\DeclareMathOperator{\II}{\ensuremath{\mathcal{I}}}
\DeclareMathOperator{\g}{\ensuremath{\boldsymbol{\gamma}}}

\newcommand{\primal}[1]{\ensuremath{G_P(#1)}\xspace}

\newcommand{\coloring}[2]{$(#1,#2)$-coloring\xspace}
\newcommand{\cmcoloring}{\coloring{c}{m}}

\newcommand{\faireditp}[1]{$#1$-\textsc{FDMC}\xspace}
\newcommand{\binfairedit}{\faireditp{2}}
\newcommand{\fairedit}{\textsc{FDMC}\xspace}

\newcommand{\dist}[2]{\text{Hamm}(#1,#2)\xspace}

\title{Matrix Editing Meets Fair Clustering: Parameterized Algorithms and Complexity}
\author {
    Robert Ganian, Hung P. Hoang, Simon Wietheger
}
\affiliations{
	Algorithms and Complexity Group, TU Wien, Austria \\
	rganian@gmail.com, \{phoang, swietheger\}@ac.tuwien.ac.at
}

\begin{document}

\maketitle

\begin{abstract}
We study the computational problem of computing a fair means clustering of discrete vectors, which admits an equivalent formulation as editing a colored matrix into one with few distinct color-balanced rows by changing at most $k$ values. While \NP-hard in both the fairness-oblivious and the fair settings, the problem is well-known to admit a fixed-parameter algorithm in the former ``vanilla'' setting. As our first contribution, we exclude an analogous algorithm even for highly restricted fair means clustering instances. We then proceed to obtain a full complexity landscape of the problem, and establish tractability results which capture three means of circumventing our obtained lower bound: placing additional constraints on the problem instances, fixed-parameter approximation, or using an alternative parameterization targeting tree-like matrices.
\end{abstract}


\section{Introduction}
\label{sec:intro}

In a typical matrix modification problem, we are given a
 matrix $\MM$ and are tasked with modifying it into some matrix $\MM'$ satisfying a specified desirable property. Matrix modification problems arise in a broad range of research contexts directly related to artificial intelligence and machine learning, prominently including recommender systems and data recovery~\cite{CandesP10,CandesR12,ElhamifarV13} 
but also occurring in, e.g., Markov inference~\cite{RothY05}  
and computational social choice~\cite{BredereckCHKNSW14}.
Matrix completion is perhaps the most classical example of matrix modification: there, certain entries in the provided matrix $\MM$ are marked as ``missing'' and the task is to complete the missing entries. A second classical example---one which will be the focus of our interests here---is matrix editing, where we are allowed to alter at most $k$ entries of a (complete) matrix $\MM$ in order to achieve the sought-after property. 

The vast majority of matrix completion and matrix editing problems are known to be \NP-hard, leading to the investigation of these problems using the more refined \emph{parameterized complexity} paradigm~\cite{CyganFKLMPPS15}.
 There, the general aim is to circumvent the intractability of problems by designing algorithms with running times which are not exponential in the whole input size, but only exponential in some well-defined integer \emph{parameters} of the input. From a complexity-theoretic perspective, we ask for which natural parameters $p$ one can obtain an algorithm solving the problem in time $f(p)\cdot n^{\bigoh(1)}$, where $f$ is a computable function and $n$ the input size; such algorithms are called \emph{fixed-parameter} and form a weaker (but still desirable) baseline of tractability than polynomial-time algorithms. 
 
The parameterized complexity of matrix completion was first investigated by Ganian, Kanj, Ordyniak and Szeider~\cite{GanianKOS18},
who targeted the two fundamental cases where $\MM'$ must adhere to an input-specified bound $r$ on either the rank, or the number of distinct rows. 
Subsequent works in the completion setting then considered a variety of different constraints on the output matrix $\MM'$~\cite{EibenGKOS21,EibenGKOS23j,GanianHKOS22,KoanaFN20,KoanaFN23}.
For matrix editing, Fomin, Golovach and Panolan~\cite{FominGP20}
studied the parameterized complexity of the two problems analogous to those considered in the completion setting~\cite{GanianKOS18}
and as their main positive result obtained a fixed-parameter algorithm for the task of editing a binary matrix to achieve at most $r$ distinct rows, parameterized by the budget $k$ on the number of altered entries (i.e., edits).\footnote{We remark that there, the bound $r$ is chosen to be on the number of columns as opposed to rows; however, the role of columns and rows is entirely symmetric.}
This task is particularly interesting, as it precisely corresponds to the classical \textsc{Binary Means Clustering} problem~\cite{kleinberg2004segmentation,OstrovskyR02}---a discrete counterpart to the means clustering that is frequently used on real-valued data in machine learning~\cite{CharikarHHVW23,MaromF19,ZhangLX20}.
Intuitively, the reason the editing and clustering tasks coincide is that each time we edit a row $\vec v$ to its final value $\vec w$, the number of edits is equal to the cost of placing a data point $\vec v$ into a cluster centered at $\vec w$; see Figure~\ref{fig:exampleME} (\emph{Top}) for an illustration.

In this article, we investigate the computational complexity of the same task of editing a matrix to achieve at most $r$ distinct rows, but in the presence of a \emph{fairness constraint}. The reason for considering fairness in this setting is directly tied to the clustering perspective, where requiring each of the clusters to be ``fair'' is equivalent to ensuring that the (at most $r$) distinct rows in $\MM'$ are ``fair''. Fair clustering has become an increasingly important research topic in recent years~\cite{Amagata24,BackursIOSVW19,BandyapadhyayFS24,DickersonEMZ23}, 
starting from the pioneering paper of Chierichetti, Kumar, Lattanzi and Vassilvitskii~\cite{Chierichetti0LV17}.
The fairness constraint we adopt here is the same as in the latter foundational work: each row is equipped with a specified color (representing an aspect of that data point that should be proportionately represented in clusters) and each cluster must admit a partitioning into \emph{fairlets}, which are minimum sets of colored elements exhibiting the same color ratio as $\MM$.
(For example, for two colors with 1:1 ratio, a fairlet contains one element of each color---see \cref{fig:exampleME} (\emph{Bottom}).) 
We remark that while this ``canonical'' fairness constraint has been used in several related works~\cite{AhmadianEK0MMPV20,BandyapadhyayFS24,Casel0SW23}, 
we also discuss possible extensions of our results to other fairness notions in Section~\ref{sec:conclusions}.

\begin{figure}[ht]
    \centering
    \includegraphics[scale=1,page=1]{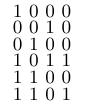}%
    \hspace{1mm}
    \includegraphics[scale=1,page=2]{graphics/example_me.pdf}%
    \hspace{1mm}
    \includegraphics[scale=1,page=3]{graphics/example_me.pdf}%

    \includegraphics[scale=1,page=1]{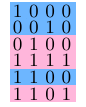}%
    \hspace{1mm}
    \includegraphics[scale=1,page=3]{graphics/example_fair_ME.pdf}%
    \hspace{1mm}
    \includegraphics[scale=1,page=4]{graphics/example_fair_ME.pdf}%
    \caption{(\textbf{Top}) A matrix (left) received 7 \underline{\textcolor{red}{edits}} (center), reducing the number of distinct rows from 6 to 2. Equivalently, the rows are partitioned into clusters (right) with centers 0000 and 1011, respectively. Hamming distances between rows and their respective center equal the number of edits in the middle.
    (\textbf{Bottom}) Example of matrix editing with fairness colors \textcolor{cb_light_blue}{blue} and \textcolor{cb_rose}{rose}. The center matrix has 2 distinct rows but is not fair. The right matrix is fair, with one cluster consisting of a single fairlet (a blue row and a rose row) and the other cluster consisting of two fairlets.}
    \label{fig:exampleME}
\end{figure}

In this work, we consider both the binary and higher-domain settings. 
We denote our general problem of interest \textsc{Fair Discrete Means Clustering} (or \fairedit) over some input-specified domain (of distinct entries); this matches the fair version of the matrix editing task where $\MM'$ must have at most $r$ distinct rows. 
\faireditp{2} is then the restriction to the binary case and is equivalent to \textsc{Binary Means Clustering} under the aforementioned fairness constraint. Formal definitions are provided in Section~\ref{sec:prelims}.


\paragraph{Contributions.}
Given the aforementioned fixed-parameter algorithm of Fomin, Golovach and Panolan~\cite{FominGP20}
for \textsc{Binary Means Clustering} parameterized by $k$, a first natural question that arises in our study is whether one can achieve an analogous result when we require the obtained clustering to be fair. As our first contribution, we provide a non-trivial reduction which rules this out under well-established complexity assumptions:

\begin{restatable}{theorem}{thmwhard}
\label{thm:whard_fixed}
   \binfairedit is \wonehard when parameterized by the fairlet size $\tilde{c}$ plus the budget $k$, even if $\MM$ already achieves the target number $r$ of distinct rows.
\end{restatable}

One may notice that Theorem~\ref{thm:whard_fixed} excludes fixed-parameter algorithms not only when parameterized by $k$ alone, but even if the parameter includes $\tilde{c}$, the fairlet size and---as a direct consequence---also the number of colors. Essentially, our reduction shows that the problem becomes intractable w.r.t.\ $k$ if all the clusters must include a balanced combination of a small number of colors.
However, what is perhaps even more remarkable is that the lower bound holds even if $\MM$ already has at most $r$ distinct rows, i.e., it also applies to instances which are trivial in the ``vanilla'' setting without fairness.

In the remainder of the paper, we provide tractability results to circumvent this strong lower bound via three different avenues: additional constraints, approximation, and alternative parameterizations.

\subparagraph{Additional Constraints.}
In Theorem~\ref{thm:k_le_c}, we show that \fairedit\ is, in fact, fixed-parameter tractable w.r.t.\ the number of edits when dealing with instances where the fairlet sizes are ``sufficiently large''. 

\begin{restatable}{theorem}{thmklec}
\label{thm:k_le_c}
	When restricted to instances with $\tilde{c} > k$, \fairedit is fixed-parameter tractable w.r.t.~$k$.
%
%
\end{restatable}

Theorem~\ref{thm:k_le_c} is fairly surprising, as it shows that the effect of the fairlet size on the problem's complexity is ``non-uniform'': instances with fairlet size $1$ precisely correspond to the vanilla setting and are hence also fixed-parameter tractable w.r.t.\ $k$.
Moreover, we also obtain a fixed-parameter algorithm for the problem when parameterized by the budget plus the bound $r$ on the number of clusters. 

\begin{restatable}{theorem}{thmkplusr}
\label{thm:k_plus_r}
    \fairedit is fixed-parameter tractable with respect to $k+r$.
\end{restatable}

These results allow us to piece together the full complexity landscape of \fairedit\ when parameterized by every combination of $k$, $r$, $\tilde{c}$, the number $c$ of colors and the domain, as we discuss at the end of Section~\ref{sec:additionalc}. Nevertheless, the proofs of these two theorems are comparatively simple and can hence be seen as a gentle introduction to the reasoning we will employ for
 the results which form the bulk of our algorithmic contributions---specifically Theorems~\ref{thm:approx} and~\ref{thm:tw}.

\subparagraph{Approximation.}
Theorems~\ref{thm:k_le_c} and~\ref{thm:k_plus_r} allow us to circumvent the aforementioned lower bound if certain conditions are met; however, a more generally applicable approach would be to ask for a fixed-parameter algorithm parameterized by $k$ alone that can compute a fair clustering which is at least approximately optimal (in the number $k$ of edits). 
Fixed-parameter approximation algorithms have found applications for a number of other clustering problems to date~\cite{BandyapadhyayFS24,GoyalJ23,ZhangCLCHF24}.

In terms of approximation, the ``vanilla'' \textsc{Binary Means Clustering} is known to admit a randomized approximation~\cite{OstrovskyR02} and also a deterministic approximation~\cite{FominGLP020}.
However, both of these algorithms require the number of clusters to be fixed in order to run in polynomial time, and moreover none of the techniques developed in the previous works can be directly applied to solve our problem of interest here.

Instead, we develop a new approach utilizing a matching-based decomposition of the \emph{edit graph}, which is a hypothetical structure capturing the modifications carried out by an optimal solution. 
By using this decomposition to establish the existence of a near-optimal and ``well-structured'' clustering, we obtain the following constant-factor approximation:

\begin{restatable}{theorem}{thmapprox}
\label{thm:approx}
\fairedit\ admits a $(5 - 3/\tilde{c})$-approximate fixed-parameter algorithm with respect to~$k$.
\end{restatable}

\subparagraph{Alternative Parameterizations.}
For our final contribution, we show that one can in fact have an exact fixed-parameter algorithm at least for \binfairedit\ under a different parameterization than the number $k$ of edits. Towards this, we consider a structural measure of the input matrix $\MM$, thus yielding exact algorithms---even for instances requiring a large number of edits---whose performance scales with how ``well-structured'' $\MM$ is. Our structural measure of choice here is the \emph{treewidth} $t$ of $\MM$, which has been successfully employed for other problems on binary matrices~\cite{EibenGKOS23,GanianHKOS22}
 but not yet in the editing setting. Essentially, $t$ measures how tree-like the interactions are between rows which share the underrepresented value (say $1$) on the same coordinate, and is obtained by measuring the \emph{treewidth} of the so-called \emph{primal graph} of $\MM$~\cite{GanianHKOS22}. 
 By developing a complex dynamic programming subroutine that not only carefully aggregates information from the previously processed parts of the input, but also anticipates the properties of the remainder of the instance, we obtain:

\begin{restatable}{theorem}{thmtw}
\label{thm:tw}
    \binfairedit is fixed-parameter tractable with respect to the treewidth of $\MM$.
\end{restatable}

We remark that Theorem~\ref{thm:tw} also yields, as a special case, an alternative parameterization that can be used to solve the ``vanilla'' clustering problem studied, e.g., by Fomin, Golovach and Panolan~\cite{FominGP20}.

\section{Preliminaries}
\label{sec:prelims}
For a positive integer $i$, we write $[i]$ for the set $\set{1, 2, \dots , i}$.
For an $m \times n$ matrix $\MM$ (i.e., a matrix with $m$ rows and $n$ columns over some arbitrary domain), and for $i\in [m], j\in [n]$, $\MM[i, j]$ denotes the entry in the $i$-th row and $j$-th column of $\MM$. 
We write $\MM[i, \star]$ for the row-vector $(\MM[i, 1], \MM[i, 2], \dots , \MM[i, n])$,
and $\MM[\star, j]$ for the column-vector $(\MM[1, j], \MM[2, j], \dots , \MM[m, j])$. 
We call an $n$-dimensional vector a \emph{type} and refer to $\MM[i, \star]$ as the type of row $i$.
\iflong
We let $\type(\MM)$ denote the set of distinct types among the rows of $\MM$ and set $\dr(\MM)=|\type(\MM)|$, i.e., $\dr(\MM)$ is the number of distinct rows in $\MM$.
\fi 
\ifshort
 We let $\dr(\MM)$ denote the number of distinct types (i.e., the number of distinct rows) of $\MM$.
\fi
We denote by $p(\MM)$ the number of distinct entries in $\MM$.
 For two types $\tau_1, \tau_2,$ let $\dist{\tau_1}{\tau_2}$ be their Hamming distance.
For two $m\times n$ matrices $\MM,\MM'$, we write $\editcost{\MM}{\MM'}$ to denote the number of entries in which they differ.

We refer to the maximal sets of pairwise identical rows in a matrix as \emph{clusters}.
For an $m$-dimensional vector $\g$, and for $i \in [m]$, $\g[i]$ denotes the $i$-th entry of the vector.
For $c,m\in \N$, a vector $\g \in [c]^m$ is a \emph{\cmcoloring}. 
The input for our problem will formally include an $m\times n$ matrix $\MM$ and a \cmcoloring $\g$; we say that the $i$-th row of $\MM$ has color $\g[i]$.
We call $\MM$ \emph{fair (w.r.t.\ $\g$)} if all its clusters are fair, that is, they each witness the same color distribution as $\g$. 
Formally, for each color $i\in [c]$, each fair cluster~$S$ contains precisely $\frac{\abs{\g}_i}{m} \cdot |S|$ rows of color $i$, where $\abs{\g}_i$ denotes the number of entries of value $i$ in $\g$.
We call a type $\tau$ $\MM$-\emph{fair} if the cluster of this type in $\MM$ is empty or fair, and $\MM$-\emph{unfair} otherwise. 
For a \cmcoloring $\g$, let $\tilde{c} = m / \gcd(\abs{\g}_1,\ldots,\abs{\g}_{\tilde{c}})$
 denote the minimum size of any fair set of rows, where $\gcd$ denotes the greatest common divisor.
A \emph{fairlet} is a fair cluster of size $\tilde{c}$, and we also refer to $\tilde{c}$ as the \emph{fairlet size}.
We are now ready to define our problem of interest.\footnote{For purely complexity-theoretic reasons, here we formalize the problem in its decision variant. However, all algorithmic results obtained in this article are constructive and can also solve the corresponding optimization task.}

\vspace{0.2cm}
\begin{problem}
      \problemtitle{Fair Discrete Means Cluster Editing (\fairedit)}
      \probleminput{$m\times n$ matrix $\MM$, \cmcoloring $\g$, positive integers $k$ and $r$}
      \problemtask{Find a $m\times n$ matrix $\MM'$ that is fair for $\g$ such that $\editcost{\MM}{\MM'} \le k$ and $\dr(\MM') \le r$.}
\end{problem}



We refer to such a matrix $\MM'$ as a \emph{solution} to the instance. 
We let \binfairedit denote the problem in the binary domain, that is, 
$p(\MM) \leq 2$.
Note that for $c=1$, every $m\times n$ matrix $\MM'$ is fair and hence in this case
\fairedit reduces to the classical Matrix Editing problem without the fairness constraint.


%

Given a matrix $\MM'$ and an instance $\II = (\MM, \g, k, r)$ of \fairedit, we define the \emph{edit graph} $\egr{\MM'}_{\II}$ as the following edge-colored edge-labeled edge-weighted directed multigraph. 
The vertex set of $\egr{\MM'}_{\II}$ is the set of types occurring in $\MM$ or $\MM'$\iflong (i.e., $V(\egr{\MM'}_{\II}) = \type(\MM) \cup \type(\MM'))$\fi.
For each row index $t \in [m]$, there is exactly one edge of $\egr{\MM'}_{\II}$ from $\MM[t,\star]$ to $\MM'[t, \star]$; this edge has label $t \in [m]$, color $\gamma[t] \in [c]$, and weight $\dist{\MM[t,\star]}{\MM'[t, \star]}$. 
We drop the subscript from $\egr{\MM'}_{\II}$ when the instance is clear from context. An illustration of an edit graph is provided in \cref{fig:example_edit_graph}.

\begin{figure}[ht]
    \centering
    \includegraphics[height=.1\textwidth,page=1]{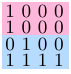}%
    \hspace{0.5cm}
    \includegraphics[height=.1\textwidth,page=2]{graphics/example_edit_graph.pdf}%
    \hspace{0.5cm}
    \includegraphics[height=.13\textwidth,page=3]{graphics/example_edit_graph.pdf}
    \caption{Example of an edit graph (right) describing the changes from a matrix $\MM$ (left) to a matrix $\MM'$ (center). Weights are printed in \textbf{bold} and labels are printed in (brackets). Solid and dashed edges represent changes in the two different fairness colors.}\label{fig:example_edit_graph}
\end{figure}

We call $\egr{\MM'}$ \emph{fair} if its edge set witnesses the same color distribution as $\g$. 
It is easy to see that $\egr{\MM'}$ is fair if and only if $\MM'$ is fair.
We say a type \emph{survives} in $\egr{\MM'}$, if it has an incoming edge in $\egr{\MM'}$.
Therefore, $\MM'$ is a solution of $\II$ if and only if $\egr{\MM'}$ is fair, at most $r$ vertices survive in $\egr{\MM'}$, and the total edge weight of $\egr{\MM'}$ is at most $k$.
Analogous to the fairness of clusters, we say that a set of edges is \emph{fair} if it witnesses the same color distribution as $\g$.

To simplify some proofs, we define the \emph{reduced edit graph} $R(\egr{\MM'})$ as the graph obtained from $\egr{\MM'}$ by removing self-loops and edge labels.
(Note that the reduced edit graph can still have multi-edges.)
Note that for a reduced edit graph $H$, there may be many matrices $\MM'$ such that $R(\egr{\MM'}) = H$; however, all of these are equivalent up to permutation of rows, and in particular either all or none are solutions.

\paragraph*{Parameterizations.}
We refer to the literature~\cite{CyganFKLMPPS15} for a formalization of parameterized complexity theory beyond the basic overview provided in Section~\ref{sec:intro}, including the notions of \emph{parameterized reductions}, \XP\ and \wone-\emph{hardness}. Note that the latter rules out fixed-parameter tractability under standard complexity assumptions.

If the domain of $\MM$ is $\set{0,1}$, the \emph{primal} graph $\primal{\MM}$ consists of vertices $v_1, \ldots, v_m$ and $v_h v_i\in E(\primal{\MM})$ if and only if there is a column $j$ such that $\MM[h,j] \neq 0$ and $\MM[i,j] \neq 0$.
In other words, we can construct $\primal{\MM}$ by first adding a vertex for each row and then, for each column of $\MM$, adding a clique among the vertices corresponding to the rows with a nonzero entry in that column.
\ifshort
The \emph{treewidth} $\tw(G)$ of a graph $G$ is a well-established measure of how ``tree-like'' it is; for instance, trees have a treewidth of $1$ while an $n$-vertex complete graph has treewidth $n-1$. For a definition of treewidth, including the notions of \emph{nice tree decompositions} and \emph{bags} which are used in the formal proof of Theorem~\ref{thm:tw}, we refer to the literature~\cite[Subsection 7.2]{CyganFKLMPPS15}.
\fi
\iflong
A \emph{nice tree decomposition} of an undirected graph $G = (V,E)$ is a pair $(\mathtt{T}, \chi)$, where $\mathtt{T}$ is a tree (whose vertices are called \emph{nodes}) rooted at a node $b_0$ and $\chi$ is a function that assigns each node $b$ a set $\chi(b) \subseteq V$ such that the following hold:
\begin{itemize}
	\item For every ${u,v} \in E$, there is a node $b$ such that $u,v \in \chi(b)$.
	\item For every vertex $v \in V$, the set of nodes $b$ satisfying $v \in \chi(b)$ forms a subtree of $\mathtt{T}$.
	\item $|\chi(\ell)| = 0$ for every leaf $\ell$ of $\mathtt{T}$ and $|\chi(b_0)| = 0$.
	\item There are only three kinds of non-leaf nodes in $\mathtt{T}$:
	\begin{itemize}
		\item \textsf{introduce} node: a node $b$ with exactly one child $b'$ such that $\chi(b) = \chi(b') \cup \set{v}$ for some vertex $v \notin \chi(b')$.
		\item \textsf{forget} node: a node $t$ with exactly one child $b'$ such that $\chi(b) = \chi(b') \setminus \set{v}$ for some vertex $v \in \chi(b')$.
		\item \textsf{join} node: a node $t$ with two children $b_1, b_2$ such that $\chi(b) = \chi(b_1) = \chi(b_2)$.
	\end{itemize}
\end{itemize}

We call each set $\chi(t)$ a \emph{bag}, and we use $\descendants(t)$ to denote the set of all vertices of $G$ which occur in the bag of $t$ or some descendant of $t$. The width of a nice tree decomposition $(\mathtt{T}, \chi)$ is the size of the largest bag $\chi(t)$ minus 1, and the \emph{treewidth} of $G$ is the minimum width of a nice tree decomposition of $G$.\fi 
\ We let the \emph{treewidth} of $\MM$ denote the treewidth of its primal graph, i.e., $\tw(\MM)= \tw(\primal{\MM})$.

\section{Fixed-Parameter Intractability of \binfairedit}
\label{sec:hardness}
We first note that the \nphardness of \textsc{Binary Means Clustering} directly transfers to \fairedit:

\begin{fact*}[\cite{Feige14b}]
    \binfairedit is \nphard even if $c=\tilde{c}=1$ and $r=2$. 
\end{fact*}

As our first contribution, we establish that---in contrast to its fairness-oblivious variant---our problem of interest does not admit a fixed-parameter algorithm under standard complexity assumptions.

\thmwhard*
\ifshort
\begin{proofsketch}
      We reduce from the \textsc{Multi-colored Clique} problem which asks, given a properly vertex-colored graph $G=(V = \{v_1,\ldots,v_{|V|}\}, E =\{e_1,\ldots,e_{|E|}\})$ with $q$ colors $[q]$, whether $G$ contains a clique of size $q$. Such a clique must contain precisely one vertex of each color, and the problem is well-known to be \wone-hard w.r.t.\ $q$~\cite{CyganFKLMPPS15}.
      	For each vertex $v$, let $\col(v)$ refer to its color.      	
      Given an instance of the \textsc{Multi-colored Clique} problem with $q\ge 4$, we construct a \binfairedit instance $(\MM, \g, \dr(\MM), k)$ as follows.
    We begin by carefully choosing a budget $k \in \bigoh(q^5)$, and produce a matrix $\MM$ with
     $m\in \bigoh(q^8 + q^6(|V| + |E|))$ rows 
     and $\bigoh(q^4+ |V|)$ columns along with a \coloring{2q}{m} $\g$ by following the description provided below.

       We identify the colors of the \textsc{Multi-colored Clique} instance with the first $q$ row colors for \binfairedit.  
%
%
	Next we list the types used in $\MM$.
	For every color $i \in [q]$, we have a set of $(q+1)$ types $\calV_{it}$ for $t \in [q] \setminus \set{i}$.
	For every pairs of colors $i,j \in [q], i < j$, we have a set of $(2q-2)$ types $\calE_{ijt}$ for $t \in [2q] \setminus \set{i,j}$.
	For each vertex $v_a$ for $a \in [|V|]$ with color $i$, we have one type $V_a$.
	Lastly, for each edge $v_a v_b$ for some $a, b \in [|V|], a < b$ with $\col(v_a) = i, \col(v_b) = j$, we have one type $E_{ab}$.
	We carefully design the types such that the pairwise Hamming distances are as listed in  \cref{table:w1_fixed}.
	
		\begin{table}[ht]
            \centering
        \begin{tabular}{r c c c c }
            \toprule
            & $\calV_{it'}$ & $\calE_{ijt'}$ & $V_{a'}$ & $E_{a'b'}$ \\
            \midrule
            $\calV_{it}$ & $4$ & $y + 4$ & $3$ & $y+4$ \\
            $\calE_{ijt}$  & $-$ & $4$ & $y+3$ & $4$ \\
            $V_{a}$  & $-$ & $-$ & $2$ & $y+1$ \\
            $E_{ab}$ & $-$ & $-$ & $-$ & $2$ \\
            \bottomrule
            \end{tabular}
        \caption{Minimum Hamming distances between distinct types\ifshort, for a carefully selected $y \in \Theta(q^3)$\fi. These are only achieved if  $\set{\col(a),\col(b)}= \set{i,j}$, $\col(a)=\col(a')$, and $\col(b)=\col(b')$; distances between other pairs of types are too large to be relevant for a hypothetical solution. \iflong Symmetric cases are marked by $-$. \fi}
        \label{table:w1_fixed}
	\end{table}
    We then create rows of each type and color such that each of the types $V_a$ and $E_{ab}$ is $\MM$-fair; at the same time, each type $\calV_{it}$ has an ``excess'' of one row of color $t$ and each type $\calE_{ijt}$ has a ``deficit'' of one row of color $t$.
      By the choices of distances between types and the ranges for $t$, we can show that the only way to make every cluster fair within the prescribed budget $k$ is to do the following.
      Each $\calV_{it}$ sends its excess row to some type $V_{h(i)}$ such that $\col(h(i)) = i$ (representing a vertex in the clique).
      Further, for each pair of colors $i<j$ in $[q]$, $E_{h(i)h(j)}$ (representing an edge in the clique) sends a row of color $j$ to $V_{h(i)}$, a row of color $i$ to $V_{h(j)}$, and each remaining row to some $\calE_{ijt}$.
      Note that for this to yield a solution, $E_{h(i)h(j)}$ has to exist for every pair of colors $i, j$ in $[q]$; that is, $v_{h(1)}, \dots, v_{h(q)}$ must form a (multi-colored) clique.
\end{proofsketch}
\fi
\iflong
\begin{proof}
      We reduce from the \textsc{Multi-colored Clique} problem, which asks, given a properly vertex-colored graph $G=(V = \{v_1,\ldots,v_{|V|}\}, E =\{e_1,\ldots,e_{|E|}\})$ with $q$ colors $[q]$, to find a multi-colored clique of size $q$ and is \wonehard with respect to $q$.
      	For each vertex $v$, we denote its color by $\col(v)$.
      Given an instance of the \textsc{Multi-colored Clique} problem with $q\ge 4$, we construct a \binfairedit instance $(\MM, \g, \dr(\MM), k)$ as follows.
      Let 
    \begin{align*}
        y \coloneqq 3q(q+1)+ 4\tbinom{q}{2}(2q-2) \text{ and }  
        k  \coloneqq 2\tbinom{q}{2} (y+1)+ y.
    \end{align*}

    
       Let $\MM$ be an $m\times n$ matrix with $m=q(2kq+1)+\binom{q}{2}(2kq-1) + 2kq|V| + 2kq|E|$ rows and $n=qy + 8q + |V|$ columns.
       Let $\g$ be a \coloring{2q}{m} matching the colors of the rows as described in the following.
       We identify the colors of the \textsc{Multi-colored Clique} instance with the first $q$ row colors for \binfairedit.  
       We consider the columns to be divided into four groups: 
       the first groups of $qy$ columns represent the $q$ colors used in the graph, with $y$ columns per color;
       the next two groups of each $4q$ columns each represent the whole set of $2q$ colors, each with $2$ columns per color; and
       the last block of $|V|$ columns represent vertices, with one column per vertex.


	We now describe the types used in $\MM$. Blocks and positions refer to the relative position within the respective group.
%
	For every color $i \in [q]$, we have a set of $(q+1)$ types $\calV_{it}$ for $t \in \set{i} \cup \set{q+1,\ldots,2q}$:
	\begin{align*}
        \calV_{it} = &
            (\underbrace{0\ldots0,\overbracket{1 \ldots1}^\labeltop{$i^{\text{th}}$ block of $y$ entries},0\ldots0}_{\text{$qy$ entries}},
            \underbrace{0\ldots0,\overbracket{1,1}^\labeltop{pos $2t$-$1$ and $2t$},0\ldots,0}_{\text{$4q$ entries}}, \\
            & \underbrace{0\ldots0}_{\text{$4q$ entries}},
            \underbrace{0\ldots0}_{\text{$|V|$ entries}}).
	\end{align*}
	For every pairs of colors $i,j \in [q], i < j$, we have a set of $(2q-2)$ types $\calE_{ijt}$ for $t \in [2q] \setminus \set{i,j}$:
	\begin{align*}
        \calE_{ijt} = &
            (\underbrace{0\ldots0,\overbracket{1 \ldots 1}^\labeltop{$i^{\text{th}}$ block of $y$ entries},0\ldots0,\overbracket{1 \ldots 1}^\labeltop{$j^{\text{th}}$ block},0\ldots0}_{\text{$qy$ entries}},
            \underbrace{0\ldots0}_{\text{$4q$ entries}}, \\
            & \underbrace{0\ldots0,\overbracket{1,1}^\labeltop{pos $2t$-$1$ and $2t$},0\ldots,0}_{\text{$4q$ entries}},
            \underbrace{0\ldots0}_{\text{$|V|$ entries}}).
	\end{align*}
	For each vertex $v_a$ for $a \in [|V|]$ with color $i$, we have one type $V_a$ as follows: 
	\[
        V_{a} = 
            (\underbrace{0\ldots0,\overbracket{1 \ldots 1}^\labeltop{$i^{\text{th}}$ block of $y$ entries},0\ldots0}_{\text{$qy$ entries}},
            \underbrace{0\ldots0}_{\text{$4q$ entries}},
            \underbrace{0\ldots0}_{\text{$4q$ entries}},
            \underbrace{0\ldots0,\overbracket{1}^\labeltop{pos $a$},0,\ldots0}_{\text{$|V|$ entries}}).
	\]
	Lastly, for each edge $v_a v_b$ for some $a, b \in |[V]|, a < b$ with $\col(v_a) = i, \col(v_b) = j$, we have one type $E_{ab}$ as follows:
	\begin{align*}
        E_{ab} = &
            (\underbrace{0\ldots0,\overbracket{1 \ldots 1}^\labeltop{$i^{\text{th}}$ block of $y$ entries},0\ldots0,\overbracket{1 \ldots 1}^\labeltop{$j^{\text{th}}$ block},0\ldots0}_{\text{$qy$ entries}},
            \underbrace{0\ldots0}_{\text{$4q$ entries}},
            \underbrace{0\ldots0}_{\text{$4q$ entries}}, \\
            & \underbrace{0\ldots0, \overbracket{1}^\labeltop{pos $a$}, 0\ldots0, \overbracket{1}^\labeltop{pos $b$}, 0\ldots0}_{\text{$|V|$ entries}}).
	\end{align*}
	
	The pairwise Hamming distances between all types are listed in \cref{table:w1_fixed}.

		\begin{table}[ht]
        \caption{Minimum Hamming distances between distinct types\ifshort, for a carefully selected $y \in \Theta(q^3)$\fi. These are only achieved if  $\set{\col(a),\col(b)}= \set{i,j}$, $\col(a)=\col(a')$, and $\col(b)=\col(b')$; distances between other pairs of types are too large to be relevant for a hypothetical solution. \iflong Symmetric cases are marked by $-$. \fi}
            \label{table:w1_fixed}
            \centering
        \begin{tabular}{r c c c c }
            \toprule
            & $\calV_{it'}$ & $\calE_{ijt'}$ & $V_{a'}$ & $E_{a'b'}$ \\
            \midrule
            $\calV_{it}$ & $4$ & $y + 4$ & $3$ & $y+4$ \\
            $\calE_{ijt}$  & $-$ & $4$ & $y+3$ & $4$ \\
            $V_{a}$  & $-$ & $-$ & $2$ & $y+1$ \\
            $E_{ab}$ & $-$ & $-$ & $-$ & $2$ \\
            \bottomrule
            \end{tabular}
	\end{table}	

      For each color $i\in [q]$ and $t \in \set{i} \cup \set{q+1,\ldots,2q}$, we create $2kq+1$ rows of type $\calV_{it}$, with $k+1$ rows of color $t$ and $k$ row each of the other $(2q-1)$ colors.
      For every pair of colors $i,j\in [q], i < j$ and $t \in [2q] \setminus \set{i,j}$, we create $2kq-1$ rows of type $\calE_{ijt}$, with $k-1$ rows of color $t$ and $k$ rows of each color other than $t$.
      For every vertex $v_a \in V$, we create $2kq$ rows of type $V_a$, with $k$ rows of each color.
      For every edge $v_a v_b \in E$, create $2kq$ rows of type $E_{ab}$, with $k$ rows of each color.
      Note that the number of occurrences of each color in $\g$ is the same, so $\tilde{c} = c = 2q$. 
      Observe that the size of the new instance is polynomial in the size of $G$ and $\tilde{c}+k$ is bounded by a computable function of $q$. 
      The statement follows as we now show that $(\MM, \g, \dr(\MM), k)$ is a \yesinstance for \binfairedit if and only if $G$ is a \yesinstance for \textsc{Multi-colored Clique}.

      Assume there is a multi-colored clique of vertices $u_1, \ldots, u_q$ in $V$ and let vertex $u_i$ have color $i$ for $i\in [q]$. 
      Let $\beta\colon [q]\rightarrow [|V|]$ be such that for each $i\in[q]$ we have $u_i = v_{\beta(i)}$. We describe a matrix $\MM'$ by defining $k$ edits made to $\MM$ to obtain $\MM'$.
      For each $i\in [q]$ and $t \in \set{i} \cup \set{q+1,\ldots,2q}$, edit a row of type $\calV_{it}$ and color $t$ into the type $V_{\beta(i)}$ representing vertex $u_i$, yielding $3$ edits per row. 
      For each edge $u_i u_j$ of the clique with $i<j$, we edit $2q$ rows of type $E_{\beta(i)\beta(j)}$ as follows.
      One row of color $i$ is edited into $V_{\beta(j)}$ and one row of color $j$ is edited into $V_{\beta(i)}$, requiring $y+1$ edits each.
      For every $t \in [2q] \setminus \set{i,j}$, we edit a row of type $E_{\beta(i)\beta(j)}$ and color $t$ into type $\calE_{ijt}$, each using $4$ edits.
      The total number of edits is then precisely $k$. 
	It is easy to check that the resulting matrix $\MM'$ is fair and witnesses $(\MM, \g, k)$ to be a \yesinstance.

	For the other direction, assume $(\MM, \g, r=\dr(\MM), k)$ is a \yesinstance of \binfairedit witnessed by a matrix $\MM'$. 
    Note that as $r=\dr(\MM)$ and each cluster in $\MM$ contains more than $k$ rows, we have that $\type(\MM) = \type(\MM')$. 
	We partition the rows of $\MM$ based on their entries in the first $qy$ columns: for $i,j\in [q], i<j$ let $\calR_i$ be the set of rows where only the $i^{\text{th}}$ block of entries has value 1 and let $\calR_{i,j}$ consist of all rows where only the $i^{\text{th}}$ and $j^{\text{th}}$ block of entries have value 1.
	For all applicable $i$ and $j$, let $T_i$ and $T_{ij}$ be the sets of types that appear in $\calR_i$ and $\calR_{ij}$, respectively. 
	For each row, we call its types in $\MM$ and $\MM'$ its \emph{relevant types}.
	We call a row whose relevant types are in different sets an \emph{interset row}, and we call a row whose relevant types are different but in the same set an \emph{intraset row}.    

	Observe that each interrow costs at least $y+1$ edits, and this number is only achieved when the relevant types of the rows are of the form $E_{ab}$ and $V_a$ (or $V_b$).
	For each color $i \in [q]$, note that the set $\calR_i$ has excess of one row of each color $t$ in $\set{i} \cup \set{q+1,\ldots, 2q}$.
	Therefore, in order for $\MM'$ to be fair, either (i) at least $q+1$ rows in $\calR_i$ obtain a type outside of $T_i$ in $\MM'$, or (ii) at least $q-1$ rows outside of $\calR_i$ obtain a type in $T_i$ in $\MM'$.
	Since each interrow costs at least $y+1$ edits, for each $\calR_i$, the lowest number of edits to make the number of rows per color the same is $(y+1)(q+1)$ in case (i) and $(y+1)(q-1)$ in case (ii).
Similarly, for each pair of colors $i$ and $j$, the set $\calR_{ij}$ has excess of one row of color $i$ and one row of color $j$.
	Hence, either (iii) at least $2q-2$ rows outside of $\calR_{ij}$ obtain a type in $T_{ij}$ in $\MM'$, or (iv) at least two rows in $\calR_{ij}$ obtain a type outside of $T_{ij}$.
    Suppose in every case, the strictly cheaper option would occur (cases (ii) and (iv)). 
    Then the total number of edits in interrows is at least 
    \[\tfrac{1}{2}(y+1)(q(q-1) + 2\tbinom{q}{2}) = 2\tbinom{q}{2} (y+1),\]
    where we divide by $2$ since potentially an interrow could be accounted for twice (in case (ii) and case (iv) when it is edited from a type in some $T_i$ to a type in some $T_{i,j}$).
    As $k < (2\tbinom{q}{2}+1) (y+1)$ and using any other case than (ii) or (iv) at any point increases the number of edits by more than $(y+1)$, we have that exactly $q-1$ rows outside of $\calR_i$ obtain a type in $T_i$ in $\MM'$, and exactly one row each of color $i$ and $j$ in $\calR_{ij}$ obtains a type outside of $T_{ij}$ for each $i,j\in[q], i<j$.

	
	Next, we count the number of edits required for intrarows.
	Consider some $i \in [q]$.
	Recall that $\calR_{i}$ receives exactly $q-1$ interrows from outside as argued above, and the unfair clusters in $\calR_{i}$ are the $q+1$ clusters of type $\calV_{it}$ in $\MM$ for $t \in \set{i} \cup \set{q+1,\ldots, 2q}$, where each such cluster has one extra row of color $t$.
	It follows that the $q-1$ received interrows have colors $[q] \setminus \set{i}$, one row for each color.
	If the extra rows in $\calR_{i}$ and the interrows have the same type $\tau$ in $\MM'$, then the smallest number of edits required for the intrarows in $\calR_i$ is $3(q+1)$, which is only achieved if $\tau$ is some $V_a$ with $\col(v_a) = i$ due to $q\ge 4$ and the Hamming distances between the types.
	Otherwise (i.e., if these rows have at least two types in $\MM'$), then there must be at least $3q-1$ intrarows in $\calR_{i}$ to make all clusters fair, so the smallest number of edits required would at least $3(3q-1) > 3(q+1)$.

	With a similar argument, we can see that for $i,j \in [q]$ with $i < j$, the number of edits required for the intrarows in $\calR_{ij}$ to make the clusters fair is at least $4(2q-2)$.
	This number is achieved, when there is some type $E_{ab}$ with $\col(a) = i$ and $\col(b) = j$ such that $2q-2$ intrarows and two interrows have type $E_{ab}$ in $\MM$. 
	
	The preceding two paragraphs imply that the minimum total number of edits required for intrarows is $3q(q+1) + 4\binom{q}{2}(2q-2) = y$.
	As argued above, the number of edits required for the interrows is at least $2\tbinom{q}{2}(y+1)$.
	Since these two numbers add up to exactly to $k$, all the conditions to achieve each minimum number of edits have to be met.
	In summary, these are
    \begin{enumerate}[label=(\alph*)]
		\item The relevant types of each interrow are of the form $V_a$ and $E_{a'b'}$ such that $a \in \set{a',b'}$.
		\item There exists a mapping $h : [q] \to [|V|]$ such that for $i \in [q]$, $\col(v_{h(i)}) = i$, and $V_{h(i)}$ is the relevant type of $q+1$ intrarows and $q-1$ interrows, and for all types in $T_i$ to be fair, these interrows must have color $[q] \setminus \{i\}$ with one row for each color.
		No interrow has a relevant type $V_a$ with $a\notin h([q])$. 
		\item For $i,j \in [q]$, $i < j$, there exist $a, b \in [|V|]$ with $a < b$, $\col(v_a) = i$, and $\col(v_b) = j$ such that $E_{ab}$ is the relevant type of $2q-2$ intrarows and two interrows.
         For all types in $T_{ij}$ to be fair, the colors of these interrows are $i$ and $j$.
	\end{enumerate}
    Consider any $i,j \in [q]$, $i < j$, let $E_{ab}$ as defined in (c), and let $\tau_i, \tau_j$ be the other relevant type of the interrows of color $i$ and $j$ in (c), respectively.
    By (a), $\tau_i, \tau_j \in \set{V_a, V_b}$ and, by (b), $\tau_i, \tau_j \in \set{V_{h(\ell)} \mid \ell\in [q]}$. 
    As $V_a$ and $V_b$ have color $i$ and $j$, respectively, this implies $\tau_i, \tau_j \in \set{V_{h(i)},V_{h(j)}}$. As by (b), $V_{h(i)}$ is not a relevant type for a row of color $i$ and the same holds for $V_{h(j)}$ and color $j$, we have $V_a = V_{h(j)} = \tau_i$ and $V_b = V_{h(i)} = \tau_j$. Thus $V_{h(i)}$ and $V_{h(j)}$ are connected by the edge $ab$ and hence $v_{h(1)}, \dots, v_{h(q)}$ induce a clique. By (b), this clique is multi-colored.
\end{proof}
\fi

\section{Fixed-Parameter Algorithms for \fairedit\ via Additional Constraints}
\label{sec:additionalc}

Our aim in this section is to identify constraints under which we can circumvent the lower bounds in Section~\ref{sec:hardness}. By the end of the section, we will in fact have obtained a comprehensive classification of the problem's complexity under the considered parameterizations.

We begin by noting that \fairedit admits a straightforward $\bigoh((mnp(\MM))^k)$-time algorithm---indeed, one can exhaustively branch over precisely which cells are edited into which value of $\MM$. 

\begin{observation}
    \fairedit is in \XP\ when parameterized by $k$.    
\end{observation}

Next, we identify the first constraint that yields fixed-parameter tractability: in particular, a fairlet size that is larger than the parameter $k$.

\thmklec*
\begin{proof}
    Every fair cluster contains at least $\tilde{c}$ rows.
	Hence, due to the fairness constraint, when $\tilde{c} > k$ we cannot create any new type in $\MM'$.
	Further, for an $\MM$-fair type, in order for it to be $\MM'$-fair, we have to remove or add either zero or at least $\tilde{c}$ rows, requiring at least $\tilde{c}$ edits.
	Therefore, as $\tilde{c} > k$, no row can be edited from or into an $\MM$-fair type.
	
	With $k$ edits, the clusters of at most $2k$ types change between an input matrix $\MM$ and a solution matrix $\MM'$.
    Thus, if more than $2k$ types are $\MM$-unfair, we can correctly reject.
    For every $\MM$-unfair type $\tau$ with a cluster $S$, note that there are at most two sizes in $\set{\tilde{c}i \mid i\in \N, |S|-k\le \tilde{c}i \le |S|+k}$ that the cluster of type $\tau$ can have in a solution $\MM'$.
    We test each of the at most $2^{2k}$ branches for all choices across all $\MM$-unfair types.
    In particular, we attempt to construct a fair reduced edit graph such that there is no new type, all $\MM$-fair types are isolated (i.e., they are incident to only self-loops in the ordinary edit graph), and the total edge weight is at most $k$.
    We do so by having colored ``half-edges'' at each $\MM$-unfair type, such that the number of incoming (or outgoing) half-edges of each color is the number of rows of that color to be added to (or removed from) its cluster to obtain a fair cluster of the branched size.
    Then we match each outgoing half-edge with an incoming half-edge of the same color and assign the combined edge the Hamming distance between the corresponding types as weight.    
	By fixing an arbitrary order of the outgoing half-edges, we can view a matching of these half-edges as a permutation of the incoming half-edges---yielding at most $\bigoh(k!)$ such matchings and allowing us to exhaustively branch over these.
	After that, for each graph, we compute its total edge weight, which can be done in $\bigO{kn}$ time.
       Clearly, the \fairedit instance has a solution if and only if at least one of the branches succeeds.
       In total, this algorithm runs in time $\bigO{2^{2k}k!\cdot kn +m}$.
\end{proof}

\iflong
\begin{corollary}
    \fairedit is fixed-parameter tractable with respect to $k$ if there is a computable function $f$ such that $\tilde{c} \in \Omega(f(mn))$.
\end{corollary}
\begin{proof}
    If $\tilde{c} > k$, then we can use the algorithm in \cref{thm:k_le_c}.
    Otherwise, we have $k \in \Omega(f(mn))$. 
    In this case, a trivial brute force algorithm testing all possible subsets of entries that are edited runs in time $\bigoh((mnp(\MM))^k)$ and hence in time $g(k)$ for some computable function $g$.
\end{proof}
\fi

Next, we show that the problem is fixed-parameter tractable with respect to $k+r$. For this, we first establish that we can assume that every \yesinstance has at least one well-structured solution\ifshort, which follows by observing that the center for any cluster can be determined by column-wise majority votes\fi.

\begin{lemma}
\label{lem:nice_solution}
    Every \yesinstance $(\MM, \g, k, r)$ of \fairedit
    is witnessed by a matrix $\MM^*$ such that for every cluster $S = \set{s_1,\ldots, s_{|S|}}$ in $\MM^*$ and for each $j\in [n]$, $\MM^*[s_1, j]$ corresponds to the winner of a majority vote of $\set{\MM[s_i, j]\mid i\in [|S|]}$, breaking ties arbitrarily. 
\end{lemma}
\iflong
\begin{proof}
    Consider any witness matrix $\MM'$ and assume there is a cluster $S$ which does not satisfy the above property in a column $j$. Let $x$ be the winner of the majority vote in $\set{\MM[s_i, j]\mid i\in [|S|]}$. Let $\MM^*$ be an $m\times n$ matrix such that $\MM^*[s_i,j] = x$ for all $i\in [|S|]$ and all remaining entries are the same as in $\MM'$. 
    Note that $\editcost{\MM}{\MM^*} \le \editcost{\MM}{\MM'}$ and $\dr(\MM^*) \le \dr(\MM')$.
    Fairness is preserved as well, which follows immediately if $S$ is still a cluster in $\MM^*$. Otherwise, there is a cluster $S'$ in $\MM$ such that $S\cup S'$ is a cluster in $\MM^*$. As the union of two fair clusters is fair, $\MM^*$ is fair. Repeating the above exhaustively yields a witness matrix $\MM^*$ satisfying the property.
\end{proof}
\fi

\thmkplusr*
\begin{proof}
    Since we can remove at most $k$ clusters with $k$ edits, if there are more than $k+r$ types in $\MM$, we correctly reject.
    If $\tilde{c} > k$, we use the algorithm in \cref{thm:k_le_c}. 
    Otherwise, we have $c \le \tilde{c} \le k$.
    
    Consider a modified reduced edit graph where there are $k/2$ new types that are undetermined (we will determine them later).
    The number of new types is at most $k/2$, because we can assume that a cluster with a new type is formed by rows of at least two types in $\MM$ (otherwise it would be cheaper to keep the previous type).
    Since $c < k$, there are at most $\bigoh((2{3k/2+r \choose 2}k)^{k})$ such graphs. We disregard graphs with outgoing edges of new types.
    For each graph, we then fix the new types according to the majority vote rule described in \cref{lem:nice_solution}.
    Finally, we check whether the resulting graph is indeed a fair reduced edit graph and whether the total edge weight is at most $k$.

    It is straightforward to see that for a \yesinstance, there exists a graph that passes all the checks above, while for a \noinstance, each graph will fail at least one check.
    We then output accordingly.
    The running time of the algorithm is at most $(k + r)^{\bigoh(k)} \cdot (n+m)^{\bigoh(1)}$, and the theorem follows.
\end{proof}

\iflong
\begin{corollary}\label{cor:k_plus_min_r}
    \fairedit is fixed-parameter tractable with respect to $k+\dr(\MM)$.
\end{corollary}
\begin{proof}
    With $k$ edits at most $\frac{k}{\tilde{c}}$ fair clusters with new types can be created. Thus, $(\MM, \g, r, k)$ is a \yesinstance if and only if $(\MM, \g, \min(r, \dr(\MM)+ \frac{k}{\tilde{c}}), k)$ is a \yesinstance, so we can use \cref{thm:k_plus_r} and assume $r \le \dr(\MM)+ \frac{k}{\tilde{c}}$.
\end{proof}
\fi

With Theorem~\ref{thm:k_plus_r}, we have an essentially full picture of the parameterized complexity of \fairedit\ w.r.t.\ $k$, $r$, $p(\MM)$, $c$ and $\tilde{c}$. Indeed, the problem is fixed-parameter tractable w.r.t.\ every superset of $\{k,r\}$. All remaining considered parameterizations which are supersets of $\{k\}$ yield \XP-tractability and \wone-hardness. Finally, \fairedit\ is para\NP-hard under every considered parameterization not covered by the first two cases.

\section{Parameterized Approximation}
\label{sec:approx}


In this section, we discuss the parameterized approximability of \fairedit.
To this end, we say that \fairedit\ admits an \emph{$\alpha$-approximate} fixed-parameter algorithm if there is a fixed-parameter algorithm that for every instance $\II=(\MM, \g, r, k)$ of \fairedit either correctly identifies that $\II$ is a \noinstance, or outputs a fair $m \times n$ matrix $\MM'$ such that $\dr(\MM')\le r$ and $\editcost{\MM}{\MM'} \le \alpha k$.
\iflong
Before discussing the $(5-3/\tilde{c})$-approximate fixed parameter algorithm in \cref{thm:approx}, we present the two following lemmas that support the proof of the theorem.
The first one is an easy consequence of the triangle inequality.

\begin{longlemma}\label{lem:fair_send_or_receive}
    Every \yesinstance $(\MM, \g, r, k)$ of \fairedit is witnessed by solution $\MM^*$ such that in $R(\egr{\MM^*})$, every $\MM$-fair type has no out-edges or no in-edges.
\end{longlemma}
\begin{proof}
	Let $\MM'$ be a witness of the \yesinstance $(\MM, \g, r, k)$.
	Suppose there exists an $\MM$-fair type that has both an in-edge and out-edge in $R(\egr{\MM'})$.
	Since $\tau$ is $\MM$- and $\MM^*$-fair, there is an out-edge $(\tau, \tau')$ in $R(\egr{\MM'})$ and an in-edge $(\tau'',\tau)$ in $R(\egr{\MM'})$ of the same color. 	
	Then we replace the edge $(\tau, \tau')$ by $(\tau, \tau)$ and $(\tau'',\tau)$ by $(\tau'',\tau')$. 
	This operation does not increase the total edge weight since the weight of the two edges changes from $\dist{\tau''}{\tau}+\dist{\tau}{\tau'}$ to $\dist{\tau''}{\tau'}+\dist{\tau}{\tau} \le \dist{\tau''}{\tau}+\dist{\tau}{\tau'}$ by the triangle inequality.
	Since the number of edges in the reduced edit graph is reduced by one, exhaustively repeating this procedure takes polynomial time.
	It is then easy to see that the resulting graph is the edit graph corresponding to a solution that meets the requirement of the lemma.
\end{proof}

The second lemma constitutes the key structural result that \cref{thm:approx} hinges upon.
\fi
\ifshort
We start with the following key technical result that our main theorem for this section hinges upon. 
\fi

\begin{lemma}\label{lem:nice_solution_approx}
    For every \yesinstance $(\MM, \g, r, k)$ of \fairedit, there is a fair matrix $\MM'$ such that \linebreak \textbf{(i)} at most $r$ vertices survive in $\egr{\MM'}$;
    \textbf{(ii)} all these surviving vertices are types of $\MM$; \linebreak
    \textbf{(iii)} each $\MM$-fair type either has no out-neighbor or has no in-neighbor and at most one \linebreak out-neighbor in $R(\egr{\MM'})$; and 
    \textbf{(iv)} the total edge weight of $\egr{\MM'}$ is at most $(5 - 3/\tilde{c})k$.
\end{lemma}
\iflong\begin{proof}\fi
\ifshort\begin{proofsketch}\fi
	Let $\MM^*$ be a solution of the \yesinstance $(\MM, \g, r, k)$, and recall that $G := \egr{\MM^*}$ may contain self-loops.
	\iflong
	Further, by \cref{lem:fair_send_or_receive}, we can assume that for every $\MM$-fair type $\tau$, in $R(G)$, $\tau$ has no out-edges or no in-edges.
	\fi
	\ifshort
	By the redirection of some edges and a simple application of the triangle inequality, we can assume that for every $\MM$-fair type $\tau$, in $R(G)$, $\tau$ has no out-edges or no in-edges.
	\fi
	For the construction of the partitions and auxiliary graphs described below, see \cref{fig:example_aux_graph} for an illustration.
	
\begin{figure}[ht]
    \centering
    \includegraphics[page=4]{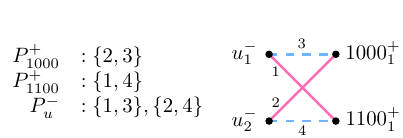}%
    \caption{For the edit graph in \cref{fig:example_edit_graph}, the partitions of edges into fair sets of size $\tilde{c} = 2$ are described on the left. (Numbers refer to labels of the edges in \cref{fig:example_edit_graph}.) Note that there are no $\MM$-fair types, and hence no partitions of the form $P^-_{\tau}$ in this example. The right side depicts the union of two auxiliary graphs for the two color classes, one with the pink solid edges and one with the blue dashed edges.}
    \label{fig:example_aux_graph}
\end{figure}
	
    As $G$ is fair, for each type $\tau$ we can partition the set of in-edges to $\tau$ (including self-loops) in $G$ into fair sets of size $\tilde{c}$. 
    Let $P_\tau^+$ denote such a partition.
    For each $\MM$-fair type $\tau$, additionally consider an arbitrary partition $P_\tau^-$ of the out-edges of $\tau$ in $G$ into fair sets of size $\tilde{c}$; here, $P_\tau^-$ is guaranteed to exist since $\tau$ is $\MM$-fair and has no out-edges or no in-edges in $R(G)$. 
    As the set of all out-edges from all $\MM$-fair types is fair, the set of all out-edges of all $\MM$-unfair types is fair as well. 
    Thus there exists a partition $P_u^-$ of the latter set into fair sets of size $\tilde{c}$.

    For every color $i\in [c]$, create an auxiliary undirected bipartite graph $H_i$ from $G$ as follows.
    For each type $\tau$, create $|P_\tau^+|$ vertices $\tau^+_1, \tau^+_2, \dots$, called \emph{$(+)$-nodes}. 
    For each $\MM$-fair type $\tau$, create $|P_\tau^-|$ vertices $\tau^-_1, \tau^-_2, \dots$, called \emph{fair $(-)$-nodes}. 
    Further, we create $|P_u^-|$ \emph{unfair $(-)$-nodes} $u^-_1, u^-_2, \dots$\iflong; note that unlike the previous kinds of nodes, these are not linked to any specific single type in $\MM$\fi.
    Consider an edge $(\tau, \rho)$ of color $i$ and some label $t \in [m]$ in $G$, i.e., $(\tau, \rho)$ intuitively represents the change in the $t^{\text{th}}$ row of $\MM$.
    Suppose this edge belongs to the $j^{\text{th}}$ part of $P_{\rho}^+$ and the $\lambda^{\text{th}}$ part of either $P_{\tau}^-$ or $P_{u}^-$, depending on whether $\tau$ is $\MM$-fair or not.
    Then we add to $H_i$ an edge of label $t$ and color $i$ between $\rho^+_j$ and either $\tau^-_{\lambda}$ or $u^-_{\lambda}$, whichever is applicable. 
    Proceed this way for all edges of color $i$.

    Let $c_i$ be the number of occurrences of the color $i$ in a fairlet with respect to $\g$.    
    Observe that $H_i$ as constructed above is bipartite with all the $(+)$-nodes as one part and all the $(-)$-nodes as the other.
    Further, every vertex in $H_i$ has degree $c_i$.
    Such a graph can be decomposed into $c_i$ perfect matchings~\cite{MR1511872}.
    Doing this for all $i \in [c]$, we obtain $\sum_{i\in [c]} c_i = \tilde{c}$ matchings in the multigraph $H$, defined as the edge-disjoint union of all $H_i$.
    We label these matchings $M_1, \dots, M_{\tilde{c}}$ in the increasing order of their total edge weights, breaking ties arbitrarily.
    Note that for each row index $t \in [m]$ exactly one edge $e$ of $H$ has label $t$, this edge has color $\g[t]$ and belongs to exactly one perfect matching.
    We assign the weight of the edge of label $t$ in $G$ to $e$.

    \ifshort
    We now modify $G$ in three phases.
    In Phase 1, our aim is to make every fair $(-)$-node have exactly one neighbor in $H$, namely its neighbor in $M_1$.
    We do this by repeatedly identifying an alternating path in the union of $M_1$ and another matching $M_j$. 
    Then we modify the edges in the matching $M_j$ in the path; most of them become identical to the edges in $M_1$, except for potentially one edge incident to an unfair $(-)$-node, whose weight may increase.
    We translate these changes in $H$ into $G$ to obtain a graph $G_1$.
    This means that in $G_1$, the set of edges from an $\MM$-fair type to another type is fair.
    
    In Phase 2, we modify $G_1$ to obtain $G_2$:
    For every $\MM$-fair type $\tau$, we redirect all its out-edges to the type $\tau_{\min}$ closest to $\tau$, among all surviving types in $G$ ($\tau_{\min}$ may be $\tau$ itself).
    Then each $\MM$-fair type has only one out-neighbor in $G_2$.
    \iflong

\fi
    In Phase 3, we aim to remove the new types (i.e., types that are not in $\MM$) from $G_2$.
    Let $\tau$ be a surviving new type in $G_2$.
    Note that $\tau$ only has in-edges.
    Let $\tau^{\text{$+$-opt}}$ be the closest in-neighbor of $\tau$ in $G$.
    We then redirect all in-edges of $\tau$ to $\tau^{\text{$+$-opt}}$.
    Exhaustively performing the procedure above until there is no surviving new type, we obtain an edit graph $G_3$.
    
    We then show that $G_3$ satisfies conditions (i)--(iv) of the lemma.
    For (iv), note that due to Phase 2, the total weight on out-edges from $\MM$-fair types may only increase in the last phase (by a factor of 2).
    For the out-edges from $\MM$-unfair types, we bound the total weight increase via a careful charging scheme and the fact that $M_1$ has the smallest total edge weight among $M_1,\ldots, M_{\tilde{c}}$.
    \fi
    \iflong
    We now proceed in three phases of modifications, after which we will argue that the desired properties (i)-(iv) hold.
    In Phase 1, our aim is to make every fair $(-)$-node have exactly one neighbor in $H$, namely its neighbor in $M_1$.
    Let $G_1$ be a graph initialized as $G$.
    We make a sequence of modifications to $H$ and $G_1$, while maintaining the following invariants:
    \begin{itemize}
    	\item[(a)] $M_1, \dots, M_{\tilde{c}}$ are $\tilde{c}$ perfect matchings in $H$ such that the edges in each matching have the same color;
    	\item[(b)] The incident edges of every vertex of $H$ form a fair set of $\tilde{c}$ edges, one edge for each of the $\tilde{c}$ perfect matchings above;
    	\item[(c)] There is a correspondence between the edges of the same label in $G_1$ and $H$. More precisely, for every row index $t \in [m]$, the edge of label $t$ in $G_1$ is of the form $(\tau, \rho)$ if and only if the edge of label $t$ in $H$ is between a $(+)$-node of $\rho$ and either a $(-)$-node of $\tau$ or an unfair $(-)$-node, depending on whether $\tau$ is $\MM$-fair.
    \end{itemize}

    
    Suppose there exists a fair $(-)$-node $v_1$ and an index $j \in \set{2, \dots, \tilde{c}}$ such that the neighbor of $v_1$ in $M_1$ is different from that in $M_j$.
    Let $e=v_2v_1$ and $e'=v_0v_1$ be the incident edges of $v_1$ in $M_1$ and $M_j$, respectively.
    Observe that the union of $M_1$ and $M_j$ is a collection of pairwise vertex-disjoint cycles, and $e$ and $e'$ are in the same cycle $C$.
    Since $v_2\neq v_0$, this cycle $C$ has length more than two.
    Note that the edges of $C$ must alternate between $M_1$ and $M_j$, in the sense that two adjacent edges of $C$ do not belong to the same perfect matching.
    
    Now, consider the maximal subpath $P = (v_0, v_1, \dots, v_{2\ell-1})$ of $C$ such that 
	for each $\zeta \in [\ell-1]$, $v_{2\zeta-1}$ is a fair $(-)$-node.
    Note that $v_{2\ell-1}$ is the only $(-)$-node that can be unfair in $P$, and if $C$ does not have any unfair $(-)$-node, then $P$ contains all vertices of $C$.
    For $i \in [2\ell]$, let $e_i$ be the edge between $v_{i-1}$ and $v_i$ in $P$ (where we define $v_{2\ell} = v_0$), and let $t_i$ be the label of this edge.
    By applying Invariant (c) on $P$, there is a unique sequence $(\tau_{2\ell}=\tau_0, \tau_1, \dots, \tau_{2\ell-1})$ of vertices in $G_1$ such that the edge of label $t_{i}$ in $G_1$ is of the form $(\tau_{i}, \tau_{i-1})$ for odd $i$ and $(\tau_{i-1}, \tau_{i})$ for even $i$. 
    We then do the following.
    In $H$, for $i \in [\ell]$, we replace the endpoints of the edge $e_{2i-1}=v_{2i-2} v_{2i-1}$ with $v_{2i-1}, v_{2i}$, while keeping all other properties of $e_{2i-1}$ intact (including its color, label, and membership in $M_j$).
    In $G_1$, for $i \in [\ell]$, we \emph{redirect} the edge $(\tau_{2i-1}, \tau_{2i-2})$ of label $t_{2i-1}$ to $\tau_{2i}$; that is, we replace its head $\tau_{2i}$ by $\tau_{2i-2}$ while keeping its color and label intact.
    Note that the edge of label $t_{2i-1}$ now has the weight $\dist{\tau_{2i-1}}{\tau_{2i}}$.
    We also update the edge of label $t_{2i-1}$ in $H$ with the same weight.
    Note that the invariants are maintained, while the neighbors of $v_1$ in $M_1$ and $M_j$ are now the same.
    Further, for any fair $(-)$-node $w$, the number of edges between it and its neighbor in $M_1$ either stays the same or is increased by one (since we potentially move its incident edge in $M_j$ in the operation above). 

    Thus, exhaustively performing the operation above results in a graph $H$ such that every fair $(-)$-node has only one neighbor.
    Consequently, using the invariants, the set of edges in $G_1$ from any $\MM$-fair type to any type is fair.

    In Phase 2, we do the following modification to $G_1$ to obtain $G_2$:
    For every $\MM$-fair type $\tau$, we redirect all its out-edges to the type $\tau_{\min}$ closest to $\tau$, among all surviving types in $G$ ($\tau_{\min}$ may be $\tau$ itself).
    Then each $\MM$-fair type has only one out-neighbor in $G_2$.

    In Phase 3, we aim to remove the new types from $G_2$.
    Let $\tau$ be a surviving new type in $G_2$.
    Note that $\tau$ only has in-edges.
    Let $\tau^{\text{$+$-opt}}$ be the closest in-neighbor of $\tau$ in $G$.
    We then redirect all in-edges of $\tau$ to $\tau^{\text{$+$-opt}}$.
    Exhaustively performing the procedure above until there is no surviving new type, we obtain an edit graph $G_3$.

	We now show that $G_3$ is a fair edit graph for some matrix $\MM'$ and $G_3$ satisfies conditions (i), (ii), (iii) of the lemma.
	Note that the invariants in Phase 1 imply that $G_1$ is a fair edit graph for some matrix $\MM_1$.
	Since in $G_1$, the set of edges from any $\MM$-fair type to any type is fair, $G_2$ is also a fair edit graph for some matrix $\MM_2$.
	Hence in Phases 2 and 3, we redirect only fair sets of edges, and each fair set is always redirected together.
	Therefore, $G_3$ is also a fair edit graph for some matrix $\MM'$.
	Further, in Phases 1 and 2, we only redirect edges to a surviving type in $G$, and consequently, the number of surviving vertices does not increase.
	In Phase~3, when we process a new type $\tau$, we remove one surviving new type (i.e., $\tau$) and add at most one more surviving type (i.e., $\tau^{\text{$+$-opt}}$).
	Hence, (i) holds for $G_3$.
	Next, since new types have no out-edges, and since $\tau^{\text{$+$-opt}}$ is an in-neighbor of $\tau$ in Phase~3, $\tau^{\text{$+$-opt}}$ is not new.
	Therefore, at the end of Phase~3, there is no new type; that is, (ii) holds for $G_3$.
	Lastly, it is easy to see that (iii) holds for $G_2$.
	In Phase~3, the only potential violation of this property is $\tau^{\text{$+$-opt}}$ for some surviving new type $\tau$.
	However, note that if $\tau^{\text{$+$-opt}}$ is $\MM$-fair, then it only has one out-neighbor (i.e., $\tau$) in $R(G_2)$.
	Therefore, after the redirection, it has no out-edges in $R(G_3)$.
	This implies that (iii) also holds for $G_3$.
	
	It remains to prove that $G_3$ satisfies (iv).
	We do this by describing a charging scheme, where a unit charged on an edge of some label $t$ indicates a number of extra edits of at most the weight of the edge of label $t$ in $G$.
	For convenience, if an edge has label $t$ such that $t^{\text{th}}$ row in $\MM$ is $\MM$-fair, we call that edge \emph{$\MM$-fair}.
	Otherwise, we call it \emph{$\MM$-unfair}.
	We have the following observation: 
	\begin{itemize}[topsep=0pt]
		\item[(*)] For an $\MM$-fair edge, if it is in a perfect matching $M_j$ other than $M_1$, it can be only be redirected at most once in Phase 1, since after a redirection, it coincides with an edge in $M_1$, and hence, in $M_j \cup M_1$, this edge belongs to a cycle of length two.
	\end{itemize}

    Consider an $\MM$-fair edge $e$.
    Then after Phase 2, $e$ is of the form $(\tau, \tau_{\min})$ for some $\MM$-fair type~$\tau$.
    By the definition of $\tau_{\min}$, the weight of this edge is at most its original weight in $G$.
    If it is changed in Phase 3, $e$ takes the form $(\tau, \tau^{\text{$+$-opt}}_{\min})$, whose weight, by the triangle inequality, is at most $\dist{\tau}{\tau_{\min}} + \dist{\tau_{\min}}{\tau^{\text{$+$-opt}}_{\min}}$.
    By the definition of $\tau^{\text{$+$-opt}}_{\min}$, the last quantity is at most $2\dist{\tau}{\tau_{\min}}$.
    This implies that we charge at most one unit to the edge of label $t$.
    For clarity later, we call this an \emph{$\alpha$-unit} of charge.

    Now consider two cases for an $\MM$-unfair edge $e$.
    In Case 1, $e$ is not modified in Phase 1, then by the same argument as above, we charge at most one unit to $e$, in case it is modified in Phase 3.
    We call this a \emph{$\beta_1$-unit} of charge.
    In Case 2, $e$ is modified in Phase 1.
    Note that the modifications to $e$ only happen when it is the last edge of some subpath $P$ of a cycle $C$.
    Each time this happens, the increase in edge weight of $e$ is at most the total edge weight along the path between $\tau_0$ and $\tau_{2\ell-2}$, by the triangle inequality.
    Observation (*) above implies that the edges between these two nodes have not been modified before we process $P$. 
    Therefore, for the increase in weight of $e$, we can charge one unit to every edge along the subpath of $P$ between $\tau_0$ and $\tau_{2\ell-2}$.
    We call this a $\beta_2$-unit of charge.
    Now, suppose at the end of Phase 1, $e$ takes the form of $(\tau_{2\ell-1},\tau_0)$.
    Then in Phase 3, it may be modified into $(\tau_{2\ell-1},\tau^{\text{$+$-opt}}_0)$, whose weight is at most $\dist{\tau_{2\ell-1}}{\tau_0} + \dist{\tau_0}{\tau^{\text{$+$-opt}}_0} \leq \dist{\tau_{2\ell-1}}{\tau_0} + \dist{\tau_0}{\tau_1}$.
    That means we will charge one more unit to the edge $(\tau_1, \tau_0)$ (which is not in $M_1$), and we call this a $\beta'_2$-unit.

    We now count the number of charges on each edge.
    Observe that an $\MM$-unfair edge can only be charged by at most one $\beta_1$-unit and nothing else.
    An $\MM$-fair edge in a perfect matching $M_j$ other than $M_1$ is charged at most one $\alpha$-unit, one $\beta_2$-unit, and one $\beta'_2$-unit by a consequence of Observation~(*) above and noting that after such an edge is charged it is part of a cycle of length 2 in the union of $M_1$ and $M_j$.
    Last, by a similar argument, an $\MM$-fair edge in $M_1$ is charged at most one $\beta_2$-unit for every perfect matching $M_j$ other than $M_1$.
    Therefore, in total, such an edge can be charged by at most $(\tilde{c}-1)$ $\beta_2$-units and one $\alpha$-unit.
        
    Let $k'$ be the total edge weight of all edges in $G$ with the same labels as those in $M_1$.
    Then the total edge weight of all other edges in $G$ is at most $k - k'$, since $G$ has total edge weight at most $k$.
    From the previous paragraph, we conclude that the total edge weight of $G_3$ is at most $(\tilde{c}+1)k' + 4(k-k') = 4k + (\tilde{c}-3)k'$.
    Recall that $M_1$ is the perfect matching with the smallest total edge weight.
    Hence, $k' \leq k / \tilde{c}$.
    That means the total edge weight of $G_3$ is at most $4k + (\tilde{c}-3)k/\tilde{c} = (5 - 3/\tilde{c})k$.
    This completes the proof that $G_3$ is the edit graph as required by the lemma.
    \fi
    \iflong
\end{proof}
\fi
\ifshort
\end{proofsketch}
\fi

\iflong 
We are now ready to prove the main result of this section.
\fi

\thmapprox*
\begin{proof}
	Consider an instance $(\MM, \g, r, k)$ of \fairedit.
    	If $\tilde{c} \ge k$, compute an exact solution using \cref{thm:k_le_c}.
    	Otherwise, if it is a \yesinstance then there exists a matrix $\MM'$ as described in \cref{lem:nice_solution_approx}.
	In particular, since the total edge weight of $\egr{\MM'}$ is at most $(5-3/\tilde{c})k$, there are at most $2(5-3/\tilde{c})k$ non-isolated vertices in $R(\egr{\MM'})$.
	Note that all $\MM$-unfair types have to be non-isolated in $R(\egr{\MM'})$.
	
    We now define a colorful variant of the problem as follows:
    Given an \fairedit instance $(\MM, \g, r, k)$ with an assignment of every $\MM$-fair type to a code\footnote{In the \emph{color coding} technique used here, these codes are usually referred to as \emph{colors}~\cite{AlonYZ95}. We employ the term \emph{code} to avoid confusion with the colors in \fairedit instances.}
    in $[2(5-3/\tilde{c})k]$, the goal is to find a fair reduced edit graph $G$ with total edge weight at most $(5-3/\tilde{c})k$ and at most $r$ surviving types such that all $\MM$-fair types of non-isolated vertices have distinct codes, or to decide that no such graph exists.

	Given a fixed-parameter algorithm with respect to $k$ for this colorful variant, we get a $(5 - 3/\tilde{c})$-approximate fixed-parameter algorithm with respect to $k$ for \fairedit, using the \emph{color coding} technique.
    In particular, it is well-established that for at most $m$ types and $2(5-3/\tilde{c})k \leq 10k$ codes one can enumerate a set of encodings of the types (an assignment of each type to a code) in time at most $k^{\bigoh(k)} m$ such that for every set $T$ of types with $|T|\le 10k$ there is at least one encoding in which all codes in $T$ are distinct~\cite[Subsections 5.2 and 5.6]{CyganFKLMPPS15}.
    Thus, given an instance of \fairedit, we can iterate through all enumerated encodings of the $\MM$-fair types and, for each of them, solve the respective instance of the colorful variant. 
    As argued above, for every \yesinstance of \fairedit at least one of the tested encodings will yield a \yesinstance of the colorful variant and output a corresponding reduced edit graph. 
    This reduced edit graph then implies the existence of a computable fair $m \times n$ matrix $\MM'$ such that $\dr(\MM')\le r$ and $\editcost{\MM}{\MM'} \le (5 - 3/\tilde{c}) k$.

	It remains to provide a fixed-parameter algorithm with respect to $k$ to solve a given colorful instance $(\MM, \g, r, k)$ with a code assignment.
	We define a \emph{template} as a directed multigraph $H = (V,E)$, where $V$ contains all $\MM$-unfair types and some other unlabeled vertices.
	Each edge of $H$ has a weight of a positive integer and a color in $[\tilde{c}]$.
	The weight of an edge between two $\MM$-unfair types must be the Hamming distance between the two types, and the total edge weight must be at most $(5-3/\tilde{c})k$.
	Each vertex is incident to at least one edge and is assigned a unique code in $[2(5-3/\tilde{c})k]$.
	Each unlabeled vertex has either at most one out-neighbor and no in-neighbor or no out-neighbors.
	
	By the definition above, a template has at most $2(5-3/\tilde{c})k$ vertices and $(5-3/\tilde{c})k$ edges.
	Since we assume $\tilde{c} < k$, and since the number of codes is at most $10k$, the number of templates is upper bounded by $k^{\bigO{k^4}}$.
	For each template, we test for the existence of an assignment of a type that is $\MM$-fair and not new to each unlabeled vertex, so that the resulting graph can then be \emph{extended} to a reduced unlabeled edit graph on $\MM$ and $\g$ (by adding suitable isolated vertices to the resulting graph).
	 In particular, for each unlabeled vertex $v$ of code $z$ we need to identify an $\MM$-fair type $\tau$ of code $z$ such that 
    (i)~for every out-edge of $v$ there is a distinct row of the same color in the cluster of type $\tau$ in $\MM$ and 
    (ii)~for every edge $(v, \tau')$ or $(\tau', v)$, the weight on the edge equals $\dist{\tau}{\tau'}$ (where $\tau'$ might be in an $\MM$-unfair type or an $\MM$-fair type assigned to some unlabeled vertex).
    
    For every unlabeled vertex, we discard all $\MM$-fair types of the respective code that do not fulfil property (i) or do not satisfy property (ii) for at least one $\MM$-unfair type.
    By the definition of the template, the induced subgraph on the unlabeled vertices is a union of pairwise vertex-disjoint stars, each of which has the edges oriented towards the center. 
    We test whether there is a suitable assignment of the remaining $\MM$-fair types to unlabeled vertices independently for every star.
    For every remaining $\MM$-fair type $\tau$ with the code of the center vertex, test whether for each leaf of the star there is at least one remaining $\MM$-fair type with respective code and whose Hamming distance to $\tau$ matches the weight of the respective edge.
    If this is successful for all unlabeled vertices, the resulting graph witnesses a \yesinstance of the colorful problem variant. 
    Otherwise, we correctly decide that $(\MM, \g, r, k)$ is a \noinstance of \fairedit.
    This process can be completed in polynomial time for each template graph, yielding an overall running time bound of $k^{\bigO{k^4}}\cdot (n+m)^{\bigoh(1)}$.
    
    It remains to show that for every \yesinstance for the colorful variant of \fairedit, at least one branch succeeds. 
    In this case, there is a matrix $\MM'$ as described in \cref{lem:nice_solution_approx}.
    Removing all isolated vertices in $\egr{\MM'}$ and the labels of all $\MM$-fair types yields a valid template and the proof follows.
\end{proof}

\section{A Treewidth-Based Fixed-Parameter Algorithm}
\label{sec:tw}
As our final contribution, we provide an alternative route towards fixed-parameter tractability for \faireditp{2}---specifically, by utilizing the structure of the input matrix $\MM$ rather than the budget $k$. Here, the restriction to binary matrices is necessary in order to facilitate a suitable definition of $\tw(\MM)$; that being said, the binary setting has also been extensively studied in the literature~\cite{FominGP20,kleinberg2004segmentation,OstrovskyR02}.

We first show that instances with a ``large'' fairlet size are trivial when parameterized by treewidth.

\begin{lemma}\label{lem:more_colors_than_tw}
    A \faireditp{2} instance $(\MM, \g, r, k)$ with $2\tw(\MM)+2 \le \tilde{c}$ is a \yesinstance if and only if $k$ is at least the number of non-zero entries in $\MM$.
\end{lemma}
\begin{proof}
    If $k$ is at least this large, changing every non-zero entry to zero will witness a \yesinstance by creating a single fair cluster containing all rows.
    For the other direction, 
    let $w= \tw(\MM) + 1$ and assume a \yesinstance is witnessed by a matrix $\MM'$ that satisfies \cref{lem:nice_solution}.
    Note that, due to the fairness requirement, every cluster in $\MM'$ contains at least $\tilde{c} \ge 2w$ rows.  
    For every column $j$ there are at most $w$ rows $i$ with $\MM[i,j]\neq 0$ as all rows with a $1$ in column $j$ would form a clique and thus cannot exceed the width $w$ of the tree decomposition. 
    Hence, in every column at least half of the rows in $S$ have value $0$, so by \cref{lem:nice_solution} and breaking ties in favor of $0$, every cluster has type $(0,\ldots,0)$. Thus, we have $k \ge \editcost{\MM}{\MM'}$, which equals the number of non-zero entries in $\MM$.
\end{proof}

\thmtw*
\iflong\begin{proof}\fi
\ifshort\begin{proofsketch}\fi
    Let $w = \tw({\MM}) + 1$ and consider a nice tree decomposition \ifshort 
$\mathtt{T}$    \fi \iflong $(\mathtt{T},\chi)$ \fi of $\primal{\MM}$ with treewidth at most $\tw({\MM})$. 
    If $\tilde{c} \ge 2w$ we solve the instance in linear time by using \cref{lem:more_colors_than_tw}.

    Otherwise, we have $\tilde{c} < 2w$ and employ a treewidth dynamic program which, intuitively, works on the following principles. 
    At each node of $\mathtt{T}$, we keep track of which \emph{states} (configurations of relevant values) are \emph{candidates} (i.e., whether they describe a valid sub-solution).
    The theorem then follows as we ensure that the number of states at each node is upper bounded by $f(w)(nm)^{\bigoh(1)}$ for a computable function $f$, the candidates at each node can be efficiently identified given the candidates of its child nodes, and any candidate at the root describes a valid solution matrix.
    
    Here, a state describes how a potential sub-solution partitions the rows (i.e., vertices in the primal graph) in the bag into clusters, the number of distinct types in the sub-solution, and the total cost of all columns which are \emph{processed} for this bag.
    Crucially, we show that for each column, there is precisely one node where a vertex representing a row with value $1$ in that column is forgotten and all remaining rows with a $1$ in that column are in the bag. 
    At this step, using \cref{lem:nice_solution}, we can already \emph{process} the column, that is, compute the total number of required edits in that column based on each possible partition of the vertices in the bag, as this partition reveals in which clusters there will be a majority of $1$s in that column (as all rows outside the bag have a $0$ there).

    To ensure fairness, states further include some additional information.
    Intuitively, whenever a new row $r$ of color $i$ is added to a bag, there are three possibilities.
    Firstly, $r$ may be part of a cluster that is of type $(0,\dots,0)$ or that includes rows from previous bags and does not overlap the current one.
    With the above observation, we know that $r$ does not share any $1$s in any column with the previous rows in the cluster, so we set all such new rows aside.
    Secondly, $r$ may be part of a cluster that overlaps the current bag.
    In this case, we need to ensure that the cluster still has space for color $i$. 
    In order to check this, for each cluster in a hypothetical solution overlapping the current bag, the state tracks its final size and how many rows of which color in that cluster were already encountered in the bag of its descendants. 
    This way, we know whether ``adding'' a new row to an existing cluster would make it irrepairably unfair (since from the final size we know how many rows of each color it should have). 
    We show that one can assume the final size of such clusters to be at most $2w$ since for larger clusters all 1s would be set to 0s by \cref{lem:more_colors_than_tw}, so the corresponding rows would be set aside.
    Lastly, $r$ may be part of a new cluster that has no rows in any of the previous or the current bags and is of type other than $(0, \dots, 0)$.
    To ensure that we do not open up too many clusters which cannot be filled later, the state also keeps track of the total (final) size of all clusters which have received at least one row from the bag or its descendants (at most $m$).    


    In the end, for a candidate, as the set of all rows is fair and the total final size of clusters that we added vertices to is at most $m$, there are sufficiently many rows of each color set aside to fill up all the remaining spots in the clusters. The other set-aside rows form a fair set as well and are placed in a cluster with type $(0,\ldots,0)$, yielding a fair solution.

\iflong
    More formally, we define a \emph{state} with respect to a node $b$ in $\mathtt{T}$ as a tuple $(\calP, s, a, q, h)$, with a partition $\calP$ of $\chi(b)$, a size function $s$, and non-negative integers $a \le m, q\le r$, and $h \le k$.
    Here, a size function $s$ maps each of the sets $P\in \calP$ to $s(P) = (s_0, s_1, \ldots, s_c)$, where $s_0 \in \set{0} \cup \set{i\tilde{c} \mid i\in \N, i\tilde{c} \le 2w}$ and the $s_z$ with $z\in [c]$ are non-negative integers such that $s_z \le \frac{|\g|_z}{m} s_0$. For convenience, we define $s_z(P)$ to refer to value $s_z$ in $s(P)$ for all $P\in\calP$ and $z \in \set{0} \cup [c]$.

    Intuitively, 
    $\calP$ partitions the rows $\chi(b)$ according to which clusters they share in a potential solution, 
    $s$ is such that $s_0$ is the total size of the corresponding cluster in the solution (note that $s_0$ can be $0$ or the size of any fair cluster up to size $2w$) and $s_1,\ldots s_c$ track how many rows of each color have been assigned to that cluster.
    The special case $s_0(P)=0$ will not necessarily correspond to all rows in $P$ sharing a cluster but simply indicate that all of their entries with value 1 will be set to 0 (though other 0 values might be set to 1). 
    Further, $a$ tracks the total number of rows that have already been accounted for, $q$ tracks how many distinct types have been created, and $h$ tracks the total cost of all \emph{processed columns} $f(b)$.
    Here, $f(b)$ is the set of columns $j$ for which there is a row $i$ represented by a vertex in $\descendants(b) \setminus \chi(b)$ such that $\MM[i,j] = 1$.

    We say that a state $(\calP, s, a, q, h)$ for a node $b$ is a \emph{candidate} for $b$ if it can be extended to a partial solution in $\descendants(b)$.
    Formally, such a state is a candidate if there is a partition $\calQ$ of the rows in $\descendants(b)$ with $|\calQ|=q$ and a function $t: \calQ \to \set{0}\cup\set{i\tilde{c} \mid i\in \N, i\tilde{c} \le 2w}$ with the following properties:
\begin{enumerate}[label=(\alph*)]
    \item The partition $\calQ$ restricted to $\chi(b)$ is $\calP$.
    \item If $P \subseteq Q$ for some $P\in\calP, Q\in \calQ$, then $s_0(P) = t(Q)$ and if $s_0(P) > 0$ then for each $z\in[c]$ the number of rows of color $z$ in $Q$ is $s_z(P)$.
	\item $\sum_{Q\in\calQ} t(Q) = a$.
	\item For every $Q\in\calQ$ with $t(Q) > 0$ and every color $z\in [c]$ there are at most $\frac{\abs{\g}_z}{m}\cdot t(Q)$ rows of color $z$ in $Q$. 
	\item $h = \sum_{j\in f(b),Q\in\calQ} \cost(t(Q), 1(j,Q))$,
\end{enumerate}
where we let $1(j,Q)$ denote the number of 1s in the $j^{\text{th}}$ column in $Q$ and define $\cost(x,y)$ as follows.
If $x = 0$ or $y \le x/2$ then $\cost(x, y) = y$. Otherwise, $\cost(x, y) = x-y$. 
Note that by using \cref{lem:nice_solution}, if $t(Q) > 0$ then $\cost(t(Q), 1(j,Q))$ is precisely the required number of edits in column $j$ in a cluster $Q^* \supseteq Q$ of size $t(Q)$ with $\MM[i,j] = 0$ for each row $i \in Q^*\setminus Q$ to give all rows in $Q^*$ the same value in column $j$.

We call a candidate $(\emptyset, \emptyset, a, q, h)$ at the root node \emph{suitable} if and only if $q \le r-1$ or $a = m$.

\begin{claim*}
    The root node $b_0$ has a suitable candidate if and only if $(\MM, \g, r, k)$ is a \yesinstance.
\end{claim*}

\begin{claimproof}
Recall that $\chi(b_0)=\emptyset$ so for any candidate we have $\calP = s = \emptyset$.
Suppose the root node has a candidate $(\emptyset, \emptyset, a, q, h)$ witnessed by some $\calQ$ and $t$. Then $\calQ$ partitions all rows in $\MM$.
Obtain a partition $\calQ^*$ from $\calQ$ by, for each $Q\in \calQ$ such that $t(Q) > 0$ and for each color $z$, moving just enough rows from sets $Q'$ with $t(Q') = 0$ to $Q$ such that $Q$ has $\frac{|\g|_z}{m}\cdot t(Q)$ rows of color $z$. 
Using (c) and (d) with $a \le m$, there are sufficiently many rows of color $z$ in sets $Q'$ with $t(Q') = 0$. 
This way, the new cluster $Q^*\supseteq Q$ consists of $t(Q)$ rows and is fair with respect to $\g$.
Further, for each $j\in [n]$, the number of edits required in the $j^{\text{th}}$ column in $Q^*$ is at most $\cost(t(Q), 1(j,Q^*)) \leq  \cost(t(Q),1(j,Q)) + 1(j, Q^* \setminus Q)$. 
Last, create one cluster from all remaining rows of sets $Q'$ with $t(Q') = 0$. As all other clusters are fair, these remaining rows form a fair cluster as well.
If we assign this cluster ${Q'}^*$ the type $(0,\ldots,0)$, the total number of edits in its rows is $\sum_{j\in [n]} 1(j, {Q'}^*)$. 
Note that this way, in each column $j$ every $Q'\in \calQ$ with $t(Q')=0$ receives precisely $1(j,Q') = \cost(0, 1(j,Q'))$ edits and each other $Q\in \calQ$ receives $\cost(j,1(jQ))$ edits.
Thus, as $f(b)= [n]$\footnote{Technically, $f(b)$ does not contain a column $j$ if $\MM[\star, j] = (0,\ldots,0)$. However, we can assume that no such column exists as otherwise it suffices to solve the instance without that column and later re-adding it and assigning each row a 0 in that column.} and by (e), the total number of edits is at most $h \le k$.
Note that there is a cluster of type $(0,\ldots,0)$ only if $a < m$ and that $q$ is the number of all clusters with at least one $1$ in the type. 
As we have $q \le r-1$ if $a < m$ and $q \le r$ else, there are at most $|\calQ^*| \le r$ distinct types.
Thus, the matrix $\MM'$ obtained by editing $\MM$ into the partitioning described by $\calQ^*$ and using the types as described above witnesses $(\MM, \g, r, k)$ to be a \yesinstance. 

For the other direction, we first note that every \yesinstance is witnessed by a matrix $\MM'$ such that all clusters are of size at most $2w$ unless they have type $(0,\ldots, 0)$: 
Consider a witness matrix $\MM'$ that satisfies \cref{lem:nice_solution} and suppose there is a cluster $S$ with $|S|>2w$ with type $\tau$.
 By the definition of treewidth, $\primal{\MM}$ has no clique of size more than $w$.
Thus no column of $\MM$ has more than $w$ non-zero entries, and hence, more than half of the rows in $C$ have value $0$ in column $j$, giving $\tau[j] = 0$ for every column.

Now suppose such a matrix $\MM'$ witnesses $(\MM, \g, r, k)$ to be a \yesinstance, where $\MM'$ consists of $q\le r$ fair clusters as described by a partition $\calQ$ of $[m]$ and requires $h \le k$ edits. Let $o$ be the number of rows of type $(0,\ldots,0)$.
Then the state $(\emptyset, \emptyset, m-o, q', h)$ with $q'=q$ if $o=0$ and $q'=q-1$ otherwise, is a candidate at the root as $\calQ$ and $t$ satisfy all required properties, where $t(Q) = 0$ if $Q$ has type $(0,\ldots,0)$ in $\MM'$ and $t(Q) =|Q|$, otherwise.
\end{claimproof}

Thus, identifying whether the root has a candidate suffices to decide the instance.
We use the following dynamic program to identify all states that are candidates at each node of $\mathtt{T}$, starting from the leaves and working up to the root. 
For each kind of node, we describe how candidates are identified under the assumption that all candidates at all child nodes are correctly identified. The correctness of the algorithm then follows by a simple inductive argument.

For convenience, we define the following notations for a bag $b$ and some row $i\in \chi(b)$. 
For some partition $\calP$ over a set of rows $U$ (either $\chi(b)$ or $\descendants(b)$), let $\calP^{- i} = \set{P\setminus \set{i}\mid P\in \calP}\setminus\set{\emptyset}$ be the partition of $U\setminus\set{i}$ as induced on $\calP$.
Further, for a partition $\calQ$ of $\descendants(b)$ and a function $t$ on $\calQ$, let $t^{-i}$ be the same function as $t$ but on $\calQ^{-i}$, that is, $t^{-i}(Q\setminus\set{i}) = t(Q)$ for all $Q\in \calQ$ except if $Q=\set{i}$ (so $t^{-i}$ is not defined on $\emptyset$).
Similarly, for a partition $\calP$ of $\chi(b)$ and a function $s$ on $\calP$ as described above, let $s^{-i}$ be the same function as $s$ but on $\calP^{-i}$, that is, $s^{-i}(P\setminus\set{i}) = s(P)$ for all $P\in \calP$, except if $P=\set{i}$ (so $s^{-i}$ is not defined on $\emptyset$).
To further account for the color of row $i$, we define $s^{-i,\g}$ to equal $s^{-i}$ except that for $P\in \calP$ with $i\in P$ we have $s_{\g[i]}^{-i,\g}(P\setminus\set{i}) = s_{\g[i]}(P)-1$ (only if $P\neq \set{i})$. 

\para{Leaves.} For a leaf $b$, we have $\descendants(b)=\emptyset$ so the only candidate is $(\emptyset, \emptyset, 0, 0, 0)$ by $\calQ = t = \emptyset$.

\para{\textsf{Introduce} Nodes.} 
For an \textsf{introduce} node~$b$ that introduces a vertex representing a row~$i$ of color $z$ into a bag $\chi(\overline{b})$, a state $(\calP, s, a, q, h)$ is a candidate of $b$  if and only if there exists a candidate  $(\calP^{-i}, s^{-i,\g}, \overline{a}, \overline{q}, h)$ of $b$, such that 
either $\overline{a} = a$, $\overline{q} = q$, and $|\calP^{-i}| = |\calP|$ 
(so $i$ has been added to an existing set) or 
$\overline{a} = a-s_0(\set{i})$,
 $\overline{q} = q-1$, 
$s_z(\set{i})=1$, $s_{z'}(\set{i})=0$ for each color $z'\neq z$,
and $\calP^{-i} = \calP \setminus \set{\set{i}}$
 (so $i$ has been added to a new set).

Suppose  $(\calP^{-i}, s^{-i,\g}, \overline{a}, \overline{q}, h)$ is witnessed to be a candidate for node $\overline{b}$ by some $\overline{\calQ}$ and $\overline{t}$.
If $\set{i}\in \calP$, let $\calQ = \overline{\calQ} \cup \set{\set{i}}$ and $t$ equal $\overline{t}$ but additionally define $t(\set{i}) = s_0(\set{i})$. 
Otherwise, there is $P\in \calP^{-i}$ such that $(P \cup \set{i})\in \calP$. 
In this case, let $\calQ$ be the same partition as $\overline{\calQ}$ but add $i$ to the set $Q \supseteq P$ and let $t$ equal $\overline{t}$ except $t(Q)$ is not defined and instead $t(Q \cup \set{i}) = \overline{t}(Q)$.
Then $\calQ$ and $t$ witness $(\calP, s, a, q, h)$ to be a candidate for $b$, where we remark that (e) holds as follows. 
Note that $\MM[i,j]=0$ for every $j\in f(\overline{b}) = f(b)$ as otherwise there is a row $i'$ represented by a vertex in $\descendants(b)\setminus\chi(b)$ with $\MM[i,j]=\MM[i',j]=1$. Then, by definition, the vertices representing $i$ and $i'$ are adjacent in $\primal{\MM}$, which contradicts them being separated by node $\overline{b}$ in the tree decomposition.
Thus, $\cost(t(\set{i}), 1(j,\set{i})) = 0$ and for every $Q \in \calQ$ we have $1(j,Q\setminus \set{\set{i}}) = 1(j,Q)$, so $\cost(t(Q\setminus \set{\set{i}}), 1(j,Q\setminus\set{\set{i}})) = \cost(t(Q), 1(j,Q))$.  

For the other direction, let $(\calP,s,a,q,h)$ be a candidate for $b$ witnessed by some $\calQ$ and $t$. Then $\calQ^{-i}$ and $t^{-i}$ witness $(\calP^{-i}, s^{-i,\g}, \overline{a}, \overline{q}, h)$ to be a candidate for $\overline{b}$, where for (e) we once more observe that row $i$ does not change the cost for any column $j\in[n]$ in any set of the partition.

\para{\textsf{Forget} Nodes.} 
For a \textsf{forget} node~$b$ that removes a vertex~$v$ from a bag $\chi(\overline{b})$, a state $(\calP, s, a, q, h)$ of $b$ is a candidate, if and only if there exists a candidate $(\overline{\calP}, \overline{s}, a, q, \overline{h})$ of $\overline{b}$, such that 
$\calP = \overline{\calP}^{-i}$,
$s = \overline{s}^{-i}$, and
$h = \overline{h} + \sum_{\overline{P} \in \overline{\calP}}\sum_{j\in f(b)\setminus f(\overline{b})} \cost(\overline{s}_0(\overline{P}), 1(j,\overline{P}))$.

Suppose $(\overline{\calP}, \overline{s}, a, q, \overline{h})$ is a candidate for $\overline{b}$ as described above witnessed by some $\calQ$ and $t$.
Then $\calQ$ and $t$ immediately also satisfy properties (a)-(d) for $(\calP, s, a, q, h)$ on $b$.
For (e), note that $f(b)\setminus f(\overline{b})$ contains precisely the columns in which $i$ has a $1$ and no other row in $\descendants(b)\setminus\chi(b)$ has a $1$.
We further have that in each such column, all rows outside $\descendants(b)$ have a $0$: otherwise, the vertex representing the row would be adjacent to the one representing $i$, which contradicts the tree decomposition forgetting $i$ before encountering the other row.
Thus, for every column $j\in f(b)\setminus f(\overline{b})$ all rows with value 1 in column $j$ are represented in $\overline{b}$. 
Consider any $Q\in \calQ$.
If $Q\cap \overline{b} = \emptyset$, then for every such column $j$ we have $1(j,Q) = \cost(t(Q)) = 0$.
Otherwise there is $P\in \calP$ such that $P\subseteq Q$. 
Then, by the above, $1(j, P) = 1(j, Q)$. 
By property (b) we get that $s_0(P) = t(Q)$ and thus $\cost(s_0(P), 1(j,P))= \cost(t(Q), 1(j,Q))$.
Hence, with $h = \overline{h} + \sum_{\overline{P} \in \overline{\calP}}\sum_{j\in f(b)\setminus f(\overline{b})} \cost(s_0(\overline{P}), 1(j,\overline{P}))$, $h$ satisfies property (e) for $(\calP, s, a, q, h)$ on $b$, and thus this state is a candidate for $b$.

For the other direction, let $(\calP, s, a, q, h)$ be a candidate for $b$ as witnessed by some $\calQ$ and $t$.
Let $Q \in \calQ$ be such that $i\in Q$.
If $Q\cap \chi(b) = \emptyset$, let $\overline{\calP} = \calP \cup \set{\set{i}}$. 
Otherwise let $\overline{\calP}$ equal $\calP$ except that the set intersecting $Q$ additionally contains $i$ in $\overline{\calP}$.
Note that in either case $\overline{\calP}^{-i} = \calP$.
Further, define $\overline{s}$ such that for each $\overline{P} \in \overline{\calP}$ we have
 $\overline{s}(\overline{P})=s(\overline{P})$ if $i\notin \overline{P}$ and otherwise let 
$\overline{s}_0(\overline{P}) = t(Q)$ and for each $z\in [c]$ let $\overline{s}_z(\overline{P})$ be the number of rows of color $z$ in $Q$. 
Then $\overline{s}^{-i} = s$. 
Let $\overline{h} = h - \sum_{\overline{P} \in \overline{\calP}}\sum_{j\in f(b)\setminus f(\overline{b})} \cost(\overline{s}_0(\overline{P}), 1(j,\overline{P}))$.
Then $\calQ$ and $t$ witness $(\overline{\calP}, \overline{s}, a, q, \overline{h})$ to be a candidate for $\overline{b}$ as follows. Properties (a), (c), and (d) hold immediately.
For (b), the only differences between $\overline{\calP}$ and $\calP$ as well as $\overline{s}$ and $s$ concern the set which includes $i$. Hence, for all other sets, (b) remains satisfied by $Q$ and $t$ and for this set, (b) holds by the definition of $\overline{s}$.
For (e), recall from the argument above that the total cost changes by precisely $\sum_{\overline{P} \in \overline{\calP}}\sum_{j\in f(b)\setminus f(\overline{b})} \cost(s_0(\overline{P}), 1(j,\overline{P}))$ between the two bags.

\para{\textsf{Join} Nodes.} 
	For a \textsf{join} node~$b$ that joins two nodes~$\overline{b}$ and~$\widetilde{b}$, a state $(\calP, s, a, q, h)$ of~$b$ is a candidate, if there exist candidates~$(\calP, \overline{s}, \overline{a}, \overline{q}, \overline{h})$ and $(\calP, \widetilde{s}, \widetilde{a}, \widetilde{q}, \widetilde{h})$ of $\overline{b}$ and $\widetilde{b}$, respectively, such that 
    $a = \overline{a} + \widetilde{a} - \sum_{P \in \calP} s_0(P)$,
    $q = \overline{q} + \widetilde{q} - |P|$, 
    $h = \overline{h} + \widetilde{h}$, 
    and for each $P\in \calP$ we have 
    $s_0(P) = \overline{s}_0(P) = \widetilde{s}_0(P)$ as well as 
    $s_z(P)$ equals $\overline{s}_z(P)+\widetilde{s}_z(P)$ minus the number of rows of color $z$ in $P$, for each color $z\in [c]$.

	Let $\overline{\calQ},\overline{t},\widetilde{\calQ},\widetilde{t}$ be the partitions and functions witnessing the candidates of $\overline{b}$ and $\widetilde{b}$, respectively.
    Let $Q$ be the joined partition of $\overline{Q}$ and $\widetilde{Q}$, that is, 
    $\calQ = 
    \set{Q\in \overline{\calQ}\mid Q\cap \descendants(\widetilde{b})=\emptyset}
    \cup
    \set{Q\in \widetilde{\calQ}\mid Q\cap \descendants(\overline{b})=\emptyset}
    \cup \set{\overline{Q} \cup \widetilde{Q} \mid \exists P\in \calP, \overline{Q}\in \overline{\calQ}, \widetilde{Q}\in \widetilde{Q}: P \subseteq \overline{Q}\land P\subseteq \widetilde{Q}}$.
    Note that $|\calQ| = |\overline{\calQ}|+|\widetilde{\calQ}|-|\calP| = q$. 
    Define function $t$ such that for each $Q\in \calQ$ we have
    $t(Q) = \overline{t}(Q)$ if $Q\cap \descendants(\overline{b})\neq 0$ and 
    $t(Q) = \widetilde{t}(Q)$, otherwise.
    Then $Q$ immediately satisfies property (a) of $(\calP, s, a, q, h)$ on~$b$ and property (b) is satisfied by the requirement on $s, \overline{s}$, and $\widetilde{s}$.
    For (c), each $Q\in \calQ$ with $Q\cap \chi(b)=\emptyset$ is accounted for once in $\overline{a} + \widetilde{a}$ and each $Q$ with $Q\cap \chi(b)\neq \emptyset$ is accounted for twice, as there is precisely one $P\in \calP$ with $P\subseteq Q$ and by (b) we have $t(Q) = s_0(P) = \overline{t}(Q) = \widetilde{t}(Q)$. Thus, (c) is satisfied by subtracting these sets counted twice:
    $a = \overline{a} + \widetilde{a} - \sum_{P \in \calP} s_0(P)$.
    Property (d) is satisfied for sets $Q\in \calQ$ with $Q\cap\chi(b)\neq\emptyset$ due to property (b). For the remaining sets, it holds since $\overline{Q},\overline{t},\widetilde{Q},\widetilde{t}$ satisfy property (d) for their respective candidate and bag.
    Last, property (e) holds since by the definition of a tree decomposition we have $\descendants(\overline{b})\cap \descendants(\widetilde{b}) = \chi(b)$, so $f(\overline{b}) \cap f(\widetilde{b}) = \emptyset$ and, by the definition of tree decompositions, there is no column such that both $\descendants(\overline{b})\setminus \chi(b)$ and $\descendants(\widetilde{b})\setminus \chi(b)$ have a row in which the entry of that column is $1$. Thus, the total cost of columns in $f(b)$ is simply the sum of the cost of columns in $f(\overline{b})$ plus the cost of columns in $f(\widetilde{b})$.
	
	For the other direction, assume $(\calP, s, a, q, h)$ is a candidate for node $b$ as witnessed by some $\calQ$ and $t$. 
	Let $\overline{\calQ} = \set{Q\cap \descendants(\overline{b})  \mid Q\in \calQ } \setminus \set{\emptyset}$ and, for each $Q \in \calQ$ with $Q\cap \descendants(\overline{b}) \neq \emptyset$, let
    $\overline{t}(Q\cap \descendants(\overline{b})) = t(Q)$.    
    Define $\widetilde{\calQ}$ and $\widetilde{t}$ analogously and note that with $\overline{q}=|\overline{\calQ}|$ and $\widetilde{q} = |\widetilde{\calQ}|$ we have $ q = \overline{q} + \widetilde{q} - |\calP|$.
    Further, let $\overline{s}, \overline{a},$ and $\overline{h}$ be such that they satisfy (b), (c), and (e) for $(\calP, \overline{s}, \overline{a}, \overline{q}, \overline{h})$ with $\overline{\calQ}$ and $\overline{t}$ on node $\overline{b}$.
    Properties (a) and (d) immediately hold for $(\calP, \overline{s}, \overline{a}, \overline{q}, \overline{h})$ on $\overline{b}$ by the definition of $\overline{Q}$ and $\overline{t}$, so $(\calP, \overline{s}, \overline{a}, \overline{q}, \overline{h})$ is a candidate for $\overline{b}$.
    The same holds for candidate $(\calP, \widetilde{s}, \widetilde{a}, \widetilde{q}, \widetilde{h})$ on node $\widetilde{b}$, where $\widetilde{s}, \widetilde{a},$ and $\widetilde{h}$ are defined analogously.
    We already established that $q = \overline{q} + \widetilde{q} - |P|$ and, due to the same reasons as used for the other direction, we have  
    $a = \overline{a} + \widetilde{a} - \sum_{P \in \calP} s_0(P)$ and $h = \overline{h} + \widetilde{h}$.
    Last, note that the requirement on $s, \overline{s},$ and $\widetilde{s}$ is satisfied for each $P\in \calP$, so all conditions for the algorithm to classify $(\calP, s, a, q, h)$ as a candidate for $b$ are satisfied.

    This concludes the description of the dynamic program. 
    We remark that the number of distinct states at each bag is in $\bigoh(w^w r  k  m  (2w)^{(c+1)w}) = f(w) (mn)^{\bigoh(1)}$ for a computable function $f$ (using that instances with $k \ge mn$ or $r\ge m$ are trivial and recalling $c \le \tilde{c} < 2w$).
Further, computing all candidates of each node by the above rules in a dynamic programming manner from the leaves to the root takes \fpt time.
\end{proof}
\fi
\ifshort
As the number of distinct states at each bag can be shown to lie in $w^{\bigoh(w^2)}krm$, the proof follows.    
\end{proofsketch}
\fi

\section{Concluding Remarks}
\label{sec:conclusions}
We remark that while our investigation concentrates on the well-established fairlet-based notion of fairness, many of our results can be lifted to different or more general variants. For example, if we assume that the instance is equipped with prescribed ranges for the proportion of each color in the clusters, then the lower bounds in Section~\ref{sec:hardness} carry over immediately to the arising more general problem. Moreover, analogous results to those obtained in Section~\ref{sec:additionalc} can be obtained using similar proof ideas; on the other hand, it is not immediately clear whether or how the more involved algorithmic ideas in Sections~\ref{sec:approx} and~\ref{sec:tw} would generalize to such a setting.

Future work can also focus on optimizing the running time bounds of
the algorithms and the lower bounds arising from the reductions in
order to pinpoint the exact fine-grained complexity of the problem
under the Exponential Time Hypothesis and/or its strong
variant~\cite{ImpagliazzoPZ01}.

Another avenue that could be explored in the future is whether the results of Section~\ref{sec:tw} can be generalized to higher-domain instances.
Finally, while this work provides a foundational analysis of the problem's complexity, it would be interesting to see how and whether the obtained insights can be used in more applied settings.

\section*{Acknowledgements}
The authors were supported by the Austrian Science Foundation (FWF, project 10.55776/Y1329).
The first author was additionally supported by FWF project 10.55776/COE12 and the second author was supported by FWF project 10.55776/ESP1136425.
We would also like to thank Frank Sommer for suggesting the decomposition of regular bipartite graphs into edge-disjoint perfect matchings that inspired the auxiliary graph in \cref{lem:nice_solution_approx}.

\bibliography{refs}


\end{document}